\documentclass{article}
\usepackage{geometry}                
\usepackage{graphicx}

\usepackage{stmaryrd}
\usepackage{amssymb}
\usepackage{amsmath, cancel, centernot }
\usepackage{amsthm}
\usepackage[mathscr]{eucal}\usepackage{mathtools}
\usepackage{epstopdf}
\title{A Gravitational Theory of the Quantum}

\author{T.N.Palmer\\ Department of Physics, University of Oxford, UK\\
tim.palmer@physics.ox.ac.uk}
\date{\today}                                        
\makeatletter
\newcommand\be{\@ifstar{\[}{\begin{equation}}}
\newcommand\ee{\@ifstar{\]}{\end{equation}}}
\newcommand\bp{\begin{pmatrix}}
\newcommand\ep{\end{pmatrix}}

\makeatother
\begin{document}
\bibliographystyle{plain}
\maketitle
\begin{abstract}
The synthesis of quantum and gravitational physics is sought through a finite, realistic, locally causal theory where gravity plays a vital role not only during decoherent measurement but also during non-decoherent unitary evolution. Invariant set theory is built on geometric properties of a compact fractal-like subset $I_U$ of cosmological state space on which the universe is assumed to evolve and from which the laws of physics are assumed to derive. Consistent with the primacy of $I_U$, a non-Euclidean (and hence non-classical) state-space metric $g_p$ is defined, related to the $p$-adic metric of number theory where $p$ is a large but finite Pythagorean prime. Uncertain states on $I_U$ are described using complex Hilbert states, but only if their squared amplitudes are rational and corresponding complex phase angles are rational multiples of $2 \pi$. Such Hilbert states are necessarily $g_p$-distant from states with either irrational squared amplitudes or irrational phase angles. The gappy fractal nature of $I_U$ accounts for quantum complementarity and is characterised numerically by a generic number-theoretic incommensurateness between rational angles and rational cosines of angles. The Bell inequality, whose violation would be inconsistent with local realism, is shown to be $g_p$-distant from all forms of the inequality that are violated in any finite-precision experiment. The delayed-choice paradox is resolved through the computational irreducibility of $I_U$. The Schr\"{o}dinger and Dirac equations describe evolution on $I_U$ in the \emph{singular} limit at $p=\infty$. By contrast, an extension of the Einstein field equations on $I_U$ is proposed which reduces \emph{smoothly} to general relativity as $p \rightarrow \infty$. Novel proposals for the dark universe and the elimination of classical space-time singularities are given and experimental implications outlined. 
\end{abstract}
\newpage
\tableofcontents
\newpage

\newtheorem{definition}{Definition}


\section{Introduction}
\label{intro}

It is widely assumed that the synthesis of quantum and gravitational physics will be achieved within the overarching framework of quantum theory. However, after 60 or more years of intense research it is unclear whether a synthesis is even remotely in sight, raising the question of whether even this basic assumption is correct. Here the problem is turned back to front: instead of seeking a quantum theory of gravity  - frequently shortened to `quantum gravity' - we ask whether the synthesis of quantum and gravitational physics might be achieved through a `gravitational theory of the quantum'.

What does such a phrase mean? For example, it has already been claimed that gravitational processes play a role during the decoherent quantum measurement problem \cite{Diosi:1989} \cite {Penrose:1989}, \cite{Percival:1995}, \cite{Ghirardi}. This would certainly be consistent with the stated aim of this paper. However, by a `gravitational theory of the quantum' is meant something more radical: a theory where gravitational processes also play a vital role in describing non-decoherent unitary evolution of the quantum-theoretic complex Hilbert state. In short, the phrase implies a synthesis of quantum and gravitational physics within the type of overarching causal geometric framework associated with general relativity theory. At first sight this would seem a nonsensical suggestion. For example, the effects of gravity are surely negligible when describing the passage of a photon through an interferometer. In addition, recent experimental tests of the Bell inequalities \cite{Gallicchio:2014} \cite{Shalm} appear to show quite conclusively that quantum physics cannot be underpinned by a theory that is fundamentally deterministic and locally causal (`locally realistic'). On these two counts alone, a `gravitational theory of the quantum' would seem to be a non-starter. 

This paper attempts to argue otherwise though an essentially finite, causal and realistic theory called `invariant set theory' \cite{Palmer:2009a} \cite{Palmer:2014}. Motivated in part by nonlinear dynamical systems theory and in part by number theory, it is proposed that the universe $U$ be considered a nonlinear deterministic system evolving precisely on a fractal invariant set $I_U$ in cosmological state space (more precisely on a measure-zero fractal-like limit cycle). This implies that the (mono-) universe evolves over multiple cosmological epochs on a compact subset of cosmological state space. Locally, this fractal structure is easily described: a single state-space trajectory segment (or history) of the universe at some $j$th fractal iterate is, at the $j+1$th iterate, a helix of $N$ trajectories, where $N$ is a large number whose inverse reflects the weakness of gravity (see Fig \ref{fractal}). In Section \ref{interfere} it is shown that if reality corresponds to an uncertain element of this helix, it can be represented probabilistically by a complex Hilbert vector, requiring both the squared amplitude and the complex phase of the vector to be rational numbers. The `Invariant Set Postulate' is formulated more formally in Section \ref{padic}, in particular that the laws of physics at their most primitive derive from the geometry of $I_U$. With this in mind, in Section \ref{Schrodinger}, the de Broglie relationships are interpreted as manifestations of the way this helical geometry of $I_U$ can determine properties of space-time. As an illustration of the latter, in Section \ref{gravity} a novel approach to distinguishing the present from the future and the past (notoriously unsolvable in conventional theory) is presented. 

The logical consistency of the analysis presented requires state-space to be described by a non-Euclidean metric $g_p$, reflecting the primacy of the geometry of $I_U$. Motivated by $p$-adic number theory, $g_p$ is defined in Section \ref{padic}, where $p=N+1$ is a large Pythagorean prime. In terms of $g_p$, two trajectory segments, one of which lies on $I_U$ and the other not, are necessarily $g_p$-distant, even though these segments may be close relative to the more familiar and intuitive Euclidean metric. The consequences of using $g_p$ as the preferred yardstick of distance in state space are transformational. For example, if a putative counterfactual state is describable by a Hilbert vector whose squared amplitude or complex phase angle is irrational, then necessarily this state is $g_p$-distant from a state lying on $I_U$ whose corresponding Hilbert vector is described by rational squared amplitudes and complex phase angles. Since conventional theories of physics based on $\mathbb R$ and $\mathbb C$ necessarily have a Euclidean state-space metric (since this metric is intrinsic to these fields' definition as completions of $\mathbb Q$), then invariant set theory cannot in particular be considered a classical theory of physics (even though it is realistic and locally causal), nor a reformulation of quantum theory.  

These issues becomes especially pertinent when considering amplitudes parametrised by the cosines of angles. In this paper an elementary number theorem is repeatedly invoked which asserts that if a given angle is a rational multiple of $2\pi$, then its cosine is irrational, with just the eight exceptions $\{0,\pi/3, \pi/2, 2\pi/3, \pi, 4\pi/3, 3\pi/2, 5\pi/3\}$. Invariant set theory is essentially a finite theory, and when an angular variable associated with a Hilbert vector is forced to take an irrational value, the corresponding state is not ontic and cannot lie on $I_U$. Making a distinction between rational and irrational descriptors would lead to an unrealistically fine-tuned theory (destroyed by arbitrarily small-amplitude noise) if the metric of state space was Euclidean. However, it leads to a robust and physically realistic theory when the metric of state space is based on the non-Euclidean $g_p$. In Section \ref{Schrodinger} this number-theoretic incommensurateness is shown to provide novel realistic explanations of two quintessential quantum phenomena: quantum complementarity in the Mach-Zehnder interferometer, and quantum spin in the sequential Stern-Gerlach apparatus. However, the most important application of such incommensurateness is in the Bell Inequality in Section \ref{nogo}. Here it is demonstrated that the Bell inequality (\ref{CHSH}) is the singular and not the smooth limit of a finite-precision form (\ref{CHSHmod}) of the inequality, in the limit where experimental precision/accuracy goes to infinity. In invariant set theory, the form of the Bell Inequality whose violation would be inconsistent with realism and local causality is undefined, and the form of the inequality that it violated experimentally is not even $g_p$-approximately close to the form needed to rule out local realism (\ref{CHSH}) \cite{Palmer:2017b}. A key element in demonstrating this result derives from the fact that experimenters cannot in principle shield their apparatuses from the uncontrollable ubiquitous gravitational waves that fill space-time. Again, this result only makes logical sense if $g_p$ is the preferred metric of state space. As discussed, although invariant set theory is deterministic, realistic and locally causal, because of $g_p$ it cannot be considered a conventional classical hidden-variable theory (whose state space is necessarily Euclidean). 

This reference to the in-principle unshieldable role of gravitational waves provides a partial reason as to why invariant set theory fits more into the mould of a gravitational theory of the quantum than a quantum theory of gravity. However, there is more than this. The whole basis for negating quantum no-go theorems are the fractal gaps in $I_U$, which contain the $g_p$-distant non-ontic states without which the theory would be classical. In Section \ref{gravity} it is argued that the existence of such fractal gaps is consistent with what is described as `information compression' at final space-time `singularities'. Of course one of the key reasons for searching for a synthesis of quantum and gravitational physics is the elimination of such space-time singularities. In Section \ref{gravity} a simple generalisation of the field equations of general relativity is described, consistent with the invariant set postulate. As well as suggesting the means by which space-time singularities might be eliminated, this generalisation also suggests a novel explanation of the dark universe. Although quantum theory arises in the \emph{singular} limit $p=\infty$, general relativity arises in the \emph{smooth} limit $p \rightarrow \infty$, suggesting invariant set theory is closer to causal deterministic general relativity theory than to quantum theory. Since quantum theory has never been found to be in discrepancy with laboratory experiments, it is concluded that experimental validation of invariant set theory will only found in situations where gravity plays an essential role. For example invariant set theory predicts that gravity is will be found to be decoherent during measurement and that there is no such particle as a graviton - indeed that the spins of elementary particle may be no larger than $\hbar$. 

\section{Realistic Interpretation of Complex Hilbert States}
\label{interfere}

The purpose of this Section is to describe how complex Hilbert vectors (and associated tensor products) can be used to provide probabilistic estimates of the uncertain state of some essentially finite system, where probability is associated with simple frequentism. In Section \ref{real} we review the use of real Hilbert vectors to describe states of classical systems under conditions of uncertainty. This is extended in Section \ref{complex} to include complex phase structure and in Section \ref{tensorreal} to include the tensor product. In Section \ref{padic} this analysis will be applied to a probabilistic description of the geometry $I_U$ in cosmological state space. 

\subsection{The Real Hilbert Vector}
\label{real}

Unit vectors in a real Hilbert space are natural mathematical entities to represent probabilistically uncertain states of a classical system. For example, let $|r\rangle$ and $|\cancel r \rangle$ denote a pair of unit orthogonal vectors representing two exclusive types of weather state: $|r\rangle$ denoting a state where it rains somewhere in London, $|\cancel r\rangle$ where it rains nowhere in London. What is the state of weather over London on Saturday week? An estimate $P_r$ of the likelihood of rain can be found by running an ensemble of weather forecasts \cite{Palmer:2006}, comprising $N$ integrations (typically 50 in practice) of a numerical weather forecast model which encodes the nonlinear Navier-Stokes and other relevant equations. The individual members of the ensemble have slightly different starting conditions, each sampled from some probability distribution of initial states, conditioned on available weather observations. As illustrated schematically in Fig 1, the state-space trajectories associated with the individual ensemble members diverge from one another under the influence of ubiquitous fluid dynamical instabilities, eventually grouping into two distinct clusters or regimes, associated, respectively, with cyclonic or anticyclonic pressure patterns over Southern England,
\begin{figure}
\centering
\includegraphics[scale=0.3]{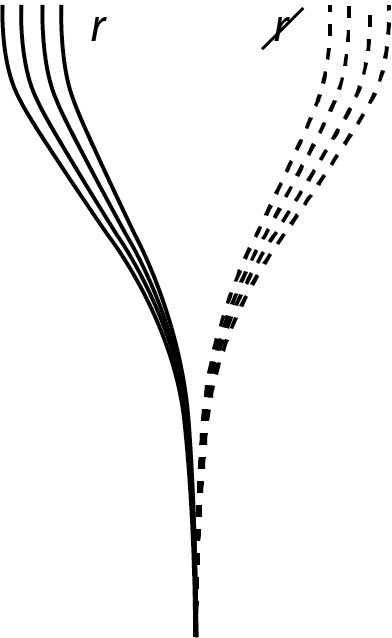}
\caption{\emph{A schematic illustration of the state-space trajectory segments of an ensemble of weather forecasts which start from almost identical initial conditions, some of which evolve to a cyclonic weather pattern associated with a wet-weather state $r$ over London, the rest evolving to a more anti-cyclonic pattern associated with a dry-weather state $\cancel r$ over London.}}
\label{weather}
\end{figure}
The ensemble of $N$ weather forecast trajectories can be represented symbolically by the bit string
\be
\label{string}
S_r=\{a_1 \; a_2 \; a_3 \ldots a_{N}\}
\ee
where $a_i \in \{r, \cancel r\}$ for all $1 \le i \le N$. If $N$ is large enough, the fraction where $a_i=r$, equal to $n_r/N$, is insensitive to $N$. In addition, if the weather model is a sufficiently accurate representation of the underlying equations then one of the forecast trajectories, say the $I$th, can be expected to shadow reality (in meteorological state space) to good accuracy.  Because $|r\rangle$ and $|\cancel r\rangle$ are unit orthogonal vectors, then by Pythagoras's theorem
\be
\label{realhilbert}
|\textrm{London weather}\rangle=\sqrt{\frac{n_r}{N}}\; |r\rangle+\sqrt{\frac{1-n_r}{N}} \;|\cancel r \rangle \nonumber
\ee
is also a unit vector and can represent the uncertain weather state over London on Saturday week, where the square of the amplitudes of the Hilbert vector give the probability of wet and dry weather respectively. Equivalently, (\ref{realhilbert}) represents the uncertain $I$th weather trajectory segment, comprising all weather states from today to Saturday, and labelled symbolically by $r$ or $\cancel r$. Only when $n_r=0$ or $n_r=N$, can we say that the weather state, or the symbolic label of the $I$th trajectory, is certain. Using the parametrisation $0 \le n_r/N=\cos^2 \theta/2 \le 1$ then:
\begin{definition} 
\label{RI1}
The unit Hilbert vector 
\be
\label{realhilbert3}
|\psi(\theta)\rangle=\cos \frac{\theta}{2} \; |a\rangle+  \sin \frac{\theta}{2} \;|\cancel a \rangle
\ee
provides a probabilistic representation of an uncertain $I$th element of $S$. By elementary frequentism, the probability that the $I$th element is equal to `$a$' is equal to $\cos^2 \theta/2 =n_r/N$. Hence it must be that $\cos \theta \in \mathbb Q $. 
\end{definition}
In both classical and quantum theory (based on the fields $\mathbb R$ and $\mathbb Q$, both completions of the rationals $\mathbb Q$), the real Hilbert vector (\ref{realhilbert3}) can correspondingly be extended to be defined for all $\cos \theta \in \mathbb R$. This extension allows one to view (\ref{realhilbert3}) as an element of an algebraically closed vector space and is of considerable analytic convenience. From a physical point of view, even though vectors with irrational values of $\cos \theta$ may not correspond to any physically realisable notion of probability, Hilbert vectors with such irrational squared amplitudes are arbitrarily close to Hilbert vectors with rational squared amplitudes, where, for vectors $x$ and $y$, the notion of `closeness' is defined by the Euclidean distance 
\be
\label{hilbertmetric}
d(x,y)=\sqrt{\langle x-y, x-y \rangle}
\ee
associated with the inner product $\langle x, y \rangle$. Since we require physical theories to be structurally stable, i.e. stable to small perturbations, this extension to irrational squared amplitudes does not and cannot itself introduce new physics and classical physics allows one to treat Hilbert vectors with irrational squared amplitudes as physically meaningful. 

\subsection{The Complex Hilbert Vector} 
\label{complex}

In quantum theory,  the complex Hilbert vector
\be
\label {hilbert4}
|\psi(\theta, \phi)\rangle=\cos \frac{\theta}{2} \; |a\rangle+ e^{i \phi} \sin \frac{\theta}{2} \;|\cancel a \rangle
\ee
 is an element of a complex one-dimensional Hilbert space $\mathscr H^1$, where the inner product $\langle x, y \rangle$ and hence (\ref{hilbertmetric}) is generalised to include a complex phase. Here it is not assumed \emph{a priori} that state space is based on $\mathscr H^1$ (and its extensions). Instead, an extension of Definition \ref{RI1} is sought that would allow complex Hilbert vectors to be interpretable probabilistically in terms of finite frequentism.  

Since the set of unit complex numbers is isomorphic to the rotation group in one dimension, we start by defining complex numbers in terms of elemental cyclic permutations on bit strings $S$. To this end, define an operator $e^{i \pi/2}$ on $S$ (given by (\ref{string}) with $a_i \in \{a, \cancel a\}$) such that
\be
\label{iS}
e^{i\pi/2}S \equiv \{\cancel a_{\frac{N}{4}+1} \;\cancel a_{\frac{N}{4}+2} \ldots \cancel a_{\frac{N}{2}}
\ \ a_1 \;a_2 \ldots a_{\frac{N}{4}}
\ \ \cancel a_{\frac{3N}{4}+1} \; \cancel a_{\frac{3N}{4}+2} \ldots \cancel a_{N}
\ \ a_{\frac{N}{2}+1}\;a_{\frac{N}{2}+2} \ldots a_{\frac{3N}{4}} \} \\ 
\\
\ee
where we have assumed $N$ is divisible by 4 (the geometric and number-theoretic implications of this assumption are discussed below) and if $a_i=a \ (\text{or}\ \cancel a)$, then $\cancel a_i = \cancel a \  (\text{or}\ a)$ respectively. For example, with $N=4$, $e^{i\pi/2} S = \{ \cancel a_2 a_1 \cancel a_4 a_3\}$. (\ref{iS}) is more transparently written as
\be
\label{iii}
e^{i \pi/2}\; S=
\begin{pmatrix}
i & 0 \\
0 & i
\end{pmatrix}
S^T 
\ee
where $S$ is treated as an $N \times 1$ row vector, and $i$ is the $N/2 \times N/2$ matrix
\be
\label{i}
i=
\begin{pmatrix}
0 & -1 \\
1 & 0
\end{pmatrix}
\ee
where $1$ and $-1$ denote the identity and negation operators respectively on blocks of $N/4$ bit-string elements. It is easily verified that 
\begin{align}
e^{i \pi} S\; &= e^{i\pi/2} \circ e^{i\pi/2} \; S = -S \nonumber \\
e^{2\pi i} S &= e^{i \pi} \circ e^{i \pi} S=S
\end{align}
which is to say that $e^{i \pi/2}$ is a `square root of minus one' operator on the bit string $S$. 

Consider now the elemental cyclical permutation operator $\zeta$
\be
\label{zeta}
\zeta \{a_1 \; a_2 \; a_3 \ldots a_{N}\}= \{a_N\; a_1\; a_2 \; a_3 \ldots a_{N-1}\}
\ee
so that $\zeta^{N}S=S$. Then if $S$ takes the particular form
\be
\label{S}
S=\{\; \underbrace{a \; a \ldots a}_{N/2} \ \underbrace{\cancel a\;  \cancel a \ldots \cancel a}_{N/2}\;\}
\ee
it is readily seen that
\be
\zeta^{N/4}\; S = e^{i \pi/2}\; S \nonumber
\ee
whence we can define the `complex exponential operator' $e^{i \phi}$ such that
\be
\label{eiphi}
e^{i \phi}\; S \equiv \zeta^n\;  S
\ee
where  
\be
\label{rational}
\frac{n}{N}=\frac{\phi}{2\pi}
\ee
for integer $n$. Since $N$ is finite, then $\phi$ must be a rational multiple of $2\pi$. Based on this, we generalise Definition \ref{RI1} so that:
\begin{definition} 
\label{RI2}
With $\cos^2 \theta/2 = n_1/N$ and $\phi/2\pi = n_2/N$, the unit complex Hilbert vector 
\be
\label{complexqubit}
|\psi(\theta, \phi)\rangle=\cos \frac{\theta}{2} \; |a\rangle+  e^{i \phi} \sin \frac{\theta}{2} \;|\cancel a \rangle
\ee
provides a probabilistic representation of a specific but unknown $I$th element of the bit string $S(\theta, \phi)=e^{i \phi} S(\theta)$ where
\be
S(\theta)=\{\; \underbrace{a \; a \ldots a}_{N \cos^2 \frac{\theta}{2}} \ \underbrace{\cancel a\;  \cancel a \ldots \cancel a}_{N \sin^2 \frac{\theta}{2}}\; \}
\nonumber
\ee
and where $e^{i \phi}$ is given by (\ref{eiphi}) and (\ref{rational}). Consistent with the Born rule, the probability that some uncertain $I$th element is an $a$ or $\cancel a$ is equal to the relevant squared amplitude of the Hilbert vector. 
\end{definition}
Consider the sphere $\mathbb S^2$ with co-latitude $\theta$ and longitude $\phi$, in relation to the Bloch Sphere of quantum theory. Then $S(\theta, \phi)$ can only be defined at points on $\mathbb S^2$ where $\cos\theta$ and $\phi/2\pi$ are rational. The north and south poles define the sequences $\{aaa \ldots a\}$ and $\{\cancel a \cancel a \cancel a \ldots \cancel a\}$ respectively and the equator contains sequences with equal numbers of symbols `$a$' and `$\cancel a$'. The following number-theoretic result is central to this paper: 
\newtheorem{theorem}{Theorem}
\begin{theorem}
\label{theorem}
 Let $\phi/\pi \in \mathbb{Q}$. Then $\cos \phi \notin \mathbb{Q}$ except when $\cos \phi =0, \pm \frac{1}{2}, \pm 1$. \cite{Niven, Jahnel:2005}
\end{theorem}
\begin{proof} 
Assume that $2\cos \phi = a/b$ where $a, b \in \mathbb{Z}, b \ne 0$ have no common factors.  Since $2\cos 2\phi = (2 \cos \phi)^2-2 \nonumber$ then 
\be
2\cos 2\phi = \frac{a^2-2b^2}{b^2} \nonumber
\ee
Now $a^2-2b^2$ and $b^2$ have no common factors, since if $p$ were a prime number dividing both, then $p|b^2 \implies p|b$ and $p|(a^2-2b^2) \implies p|a$, a contradiction. Hence if $b \ne \pm1$, then the denominators in $2 \cos \phi, 2 \cos 2\phi, 2 \cos 4\phi, 2 \cos 8\phi \dots$ get bigger without limit. On the other hand, if $\phi/\pi=m/n$ where $m, n \in \mathbb{Z}$ have no common factors, then the sequence $(2\cos 2^k \phi)_{k \in \mathbb{N}}$ admits at most $n$ values. Hence we have a contradiction. Hence $b=\pm 1$ and $\cos \phi =0, \pm\frac{1}{2}, \pm1$. 
\end{proof}
This leads to a type of number-theoretic incommensurateness which underpins invariant set theory. Consider the Hilbert vector
\be
\label{position}
|\psi_1\rangle = \frac{1}{\sqrt 2}(|a\rangle + e^{i \phi} |\cancel a\rangle)
\ee
which under the Hadamard transform
\be
\label{Hadamard}
U_H=\frac{1}{\sqrt 2} \bp  1 &1 \\ 1& -1  \ep,
\ee 
is mapped to
\be
\label{momentum}
|\psi_2\rangle= U_H |\psi_1\rangle 
= \cos\frac{\phi}{2}|a\rangle -i \sin\frac{\phi}{2} |\cancel a\rangle
\ee
modulo a global phase. Theorem \ref{theorem} implies that, for a general $\phi$ (i.e. not one of the exceptions in Theorem \ref{theorem}), if $|\psi_1\rangle$ is an uncertain representation of the $I$th element of the corresponding bit string, then $|\psi_2\rangle$ does not and \emph{vice versa}. 

In the discussion above, the operators $e^{i\phi}$ and $e^{i \phi'}$ can be treated as if they were unit complex numbers under multiplication because if both $\phi$ and $\phi'$ are rational multiples of $2 \pi$ then so is $\phi+\phi'$ and $e^{i\phi} e^{i \phi'}= e^{i(\phi+\phi')}$ is defined if $N$ is large enough. However, because of Theorem \ref{theorem} it is not permissible to treat $e^{i\phi}$ and $e^{i \phi'}$ as if they were complex numbers under addition. If they could be so treated then they would satisfy the identity
\be
\label{noadd}
\frac{1}{2}(e^{i \phi}+e^{i \phi'})=e^{i \frac{\phi+\phi'}{2}} \cos \frac{\phi-\phi'}{2}
\ee
However, since $\phi-\phi'$ is a rational multiple of $2\pi$, then $\cos (\phi - \phi')/2$ is (in general) not. That is to say, a complex Hilbert vector of the form
$|a\rangle + (e^{i\phi}+e^{i\phi'})|\cancel a\rangle$,
where $\phi$ and $\phi'$ are rational multiples of $2\pi$, in general does not correspond to a bit string and therefore has no ontic interpretation in the emerging theoretical framework. 

Before concluding, we note that the $N \times N$  `square root of minus one' operator $e^{i\pi/2}$, (\ref{iii}) is easily generalised to include the following $N \times N$ operators
\be
\label{quaternions}
\mathbf E_1=
\begin{pmatrix}
i & 0\\
0 & -i
\end{pmatrix};\ \
\mathbf E_2=
\begin{pmatrix}
0 & i\\
i & 0
\end{pmatrix};\ \ 
\mathbf E_3=
\begin{pmatrix}
0 & -1\\
1 & 0
\end{pmatrix} 
\ee
which satisfy the rules for quaternion multiplication
\be
\mathbf E^2_1=\mathbf E^2_2 = \mathbf E^2_3=\mathbf E_1 \mathbf E_2\mathbf E_3=-\mathbf 1
\ee
and from which the Pauli spin matrices and Dirac matrices are readily defined.

\subsection{The Hilbert Tensor Product}
\label{tensorreal}

Generalising the discussion above, the observed state of the universe can be considered to comprise an uncertain $I$th element of some Cartesian product $S_a \times S_b \times S_c \ldots$. In this Section, attention is focussed on the Cartesian product $S_a \times S_b$ of two correlated $N$-element bit strings. The generalisation to multiple bit strings is simply achieved by induction. 

A general `2-qubit' complex Hilbert state is given by
\be
\label{tensor}
|\psi_{ab}\rangle= \kappa_0 |a\rangle |b \rangle + \kappa_1 e^{i \chi_1} |a\rangle |\cancel{b} \rangle + \kappa_2 e^{i \chi_2}|\cancel{a}\rangle |b \rangle + \kappa_3 e^{i \chi_3}|\cancel{a}\rangle |\cancel{b} \rangle, 
\ee
where $\kappa_i, \chi_i \in \mathbb R$ and $\kappa_0^2+\kappa_1^2+\kappa_2^2+\kappa_3^2=1$. Start by setting the complex phases $\chi_i$ to zero. It is then straightforward to interpret (\ref{tensor}) classically, i.e. as an uncertain selection of some pair of elements $\{a_I, b_I\}$ from the two potentially correlated bit strings
\begin{align}
\label{twobs}
S_a&=\{a_1 \; a_2 \ldots a_{N}\} \nonumber \\
S_b&=\{b_1 \; b_2 \ldots b_{N}\} 
\end{align} 
where $a_i \in \{a, \cancel a\}$, $b_i \in \{b, \cancel b\}$. The correlation between the two bit strings can be represented graphically in two equivalent ways in Fig \ref{tensorfig}. 
\begin{figure}
\centering
\includegraphics[scale=0.4]{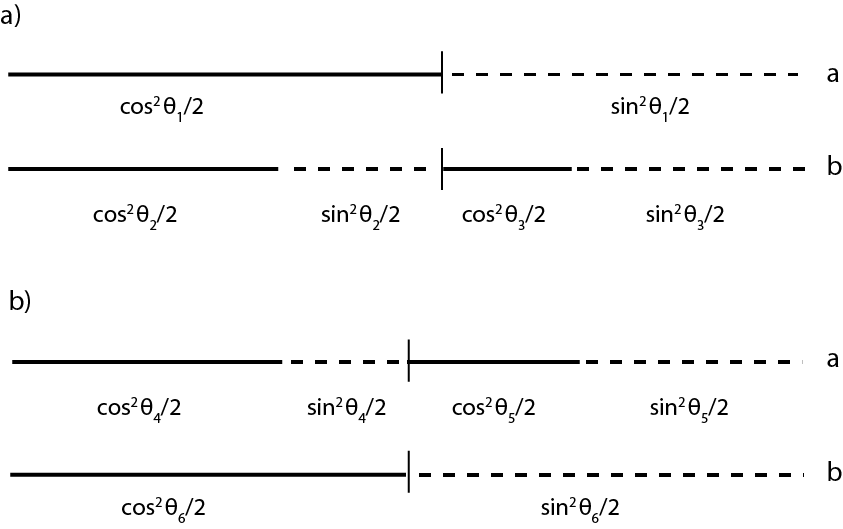}
\caption{\emph{Graphical illustrations of two equivalent definitions of the bit strings $S_a$ (top lines)and $S_b$ (bottom lines) corresponding to a 2-qubit Hilbert state. Solid lines illustrate sub-strings of bits which are either $a\;$s or $b\;$s; dashed lines refer to sub-strings of elements which are either $\cancel a\;$s or $\cancel b\;$s. }}
\label{tensorfig}
\end{figure}
For example, with reference to the first and second lines in Fig \ref{tensorfig}a, the first $N \cos^2 \theta_1/2 \; \cos^2 \theta_2/2$ elements of $S_a \times S_b$ are $(a,b)$s, the next $N \cos^2 \theta_1/2 \; \sin^2 \theta_2/2$ elements are $(a, \cancel b)$s, the next $N \sin^2 \theta_1/2 \; \cos^2 \theta_3/2$ elements are $(\cancel a, b)$s and the final $N\sin^2 \theta_1/2 \; \sin^2 \theta_3/2$ elements are $(\cancel a, \cancel b)$s. The uncertain pair of elements $(a_I, b_I)$ can therefore be represented probabilistically by the real-number tensor product
\be
\label{2qubit1}
|\psi_{ab}\rangle=
\cos \frac{\theta_1}{2}|a\rangle
|\psi_b(\theta_2)\rangle
+ \sin \frac{\theta_1}{2} |\cancel{a}\rangle |\psi_b(\theta_3)\rangle
\ee
where $|\psi_b(\theta_2)\rangle$ and $|\psi_b(\theta_3)\rangle$ are given by (\ref{realhilbert3}) over bit strings of length $N \cos^2 \theta_1/2$ and $N \sin^2 \theta_1/2$ respectively. This is equivalent to (\ref{tensor}) providing 
\be
\label{kappa}
\cos\frac{\theta_1}{2}\cos\frac{\theta_2}{2}=\kappa_0; \ \ \ \ 
\cos\frac{\theta_1}{2} \sin\frac{\theta_2}{2}=\kappa_1; \ \ \ \ 
\sin\frac{\theta_1}{2}\cos\frac{\theta_3}{2}=\kappa_2; \ \ \ \ 
\sin\frac{\theta_1}{2}\sin\frac{\theta_3}{2}=\kappa_3
\ee
The frequency of occurrence of the pairs $(a,b)$, $(a, \cancel b)$, $(\cancel a, b)$ and $(\cancel a, \cancel b)$ is equal to $\kappa^2_0$, $\kappa^2_1$, $\kappa^2_2$ and $\kappa^2_3$ respectively. 

Alternatively, $S_a$ and $S_b$ can be defined with reference to Fig \ref{tensorfig}b leading to the tensor product
\be
\label{2qubit2}
|\psi_{ab}\rangle=
\cos \frac{\theta_6}{2} |\psi_a(\theta_4)\rangle  |b\rangle+ \sin \frac {\theta_6}{2} 
|\psi_a(\theta_5)\rangle |\cancel b\rangle
\ee
which is also equivalent to (\ref{tensor}) providing
\be
\label{firsttosecond}
\cos \frac{\theta_4}{2} \cos \frac{\theta_6}{2} =\kappa_0; \ \ \ \ 
\sin \frac{\theta_4}{2} \cos \frac{\theta_6}{2} = \kappa_2; \ \ \ \ 
\cos \frac{\theta_5}{2} \sin \frac{\theta_6}{2}= \kappa_1; \ \ \ \ 
\sin \frac{\theta_5}{2} \cos \frac{\theta_6}{2}= \kappa_3 \nonumber \\
\ee

In Section \ref{complex} complex phase operators were introducted as cyclic permutation operators on bit strings. Here also the three phase degrees of freedom in the 2-qubit tensor product are introduced through the cyclical permutation operators $\zeta$. Referring first to Fig \ref{tensorfig}a, note that the elements of the two strings $S_a$ and $S_b$ can be cyclically permuted together without affecting the correlations between them. Use a phase angle $\phi_1$ to describe this degree of freedom. Similarly, one can cyclically permute the first $N \cos^2 \theta_1/2$ elements of $S_b$, or the final $N \sin^2 \theta_1/2$ elements of $S_b$, without affecting the overall correlations between $S_a$ and $S_b$. Use the phase angles $\phi_2$ and $\phi_3$ to describe these degrees of freedom, respectively. With this we finally can write:
\begin{definition}
With $\cos \theta_i \in \mathbb Q$, and $\phi_i/2\pi \in \mathbb Q$, $1 \le i \le 6$, the complex Hilbert tensor product 
\begin{align}
\label{2qubit11}
|\psi_{ab}\rangle=&
\cos \frac{\theta_1}{2}|a\rangle
|\psi_b(\theta_2, \phi_2)\rangle
+e^{i\phi_1} \sin \frac{\theta_1}{2} |\cancel{a}\rangle |\psi_b(\theta_3,\phi_3)\rangle \\
=&
\cos \frac{\theta_6}{2} |\psi_a(\theta_4, \phi_4)\rangle  |b\rangle+ e^{i \phi_6}\sin \frac {\theta_6}{2} 
|\psi_a(\theta_5, \phi_5)\rangle |\cancel b\rangle
\end{align}
can, for large enough N, provide a probabilistic representation of a specific but unknown $I$th element of $S_a \times S_b$. 
\end{definition}

\section{Fractal Invariant Sets and State Space Metrics}
\label{padic}

Recall that in $\mathscr{H}^1$ the distance between two vectors $x_{\theta \phi}$ and $x_{\theta' \phi'}$, where $\cos \theta, \phi/2\pi  \in \mathbb Q$ and  $\cos \theta', \phi'/2\pi  \notin \mathbb Q$, can be made as small as one likes, providing $\theta'$ is sufficiently close to $\theta$ and $\phi'$ sufficiently close to $\phi$. This suggests that in terms of the Hilbert Space metric, a physical theory whose states depend ontologically on the conditions  $\cos \theta \in \mathbb Q, \; \phi/2\pi \in \mathbb Q$ above would be considered unacceptably fine tuned, and therefore not structurally robust to small perturbations. This observation might appear to deal a fatal blow to the theory being developed here where states do depend ontologically on the rationality or otherwise of these cosines and phase angles. However, it is possible to approach the problem from a different angle. Consider the set $\mathscr{H}^1_{\text{Rat}}$ of complex Hilbert vectors for which $\cos \theta$ and $\phi/2\pi$ are rational. Since $\mathscr{H}^1_{\text{Rat}}$ is not an algebraically closed inner-product space, then there is no \emph{a priori} necessity to measure distances between vectors using the Hilbert Space inner product. The question then arises: Does there exist an alternate inequivalent metric where a vector with an irrational squared amplitudes is necessarily arbitrarily close to a vector with a rational amplitude squared. Here one can appeal to Ostrowsky's theorem from number theory \cite{Katok}: if there is such a metric it must be somehow related to the $p$-adic metric (where $p$ typically stands for a prime number). In general terms, $p$-adic numbers are to fractal geometry as real numbers are to Euclidean geometry. This suggests that it may be possible to develop a robust theory where Hilbert vectors with rational and irrational squared amplitudes are necessarily distant from one another, based on the primacy of some assumed fractal geometry in state space. The development of a theory based on a fractal dynamically invariant set in state space, with corresponding non-Euclidean metric, is developed below. 

\subsection{The Invariant Set Postulate}
\label{isetp}

We begin with theInvariant Set Postulate \cite{Palmer:2009a} \cite{Palmer:2014}:
\newtheorem* {guess}{Invariant Set Postulate}
\begin{guess}
\label{ISP}
\item[(1)] The universe $U$ can be considered a deterministic dynamical system whose states evolves precisely on some measure-zero invariant set $I_U$ in cosmological state space.  
\item[(2)] The laws of physics, at their most primitive, describe the geometry of $I_U$.
\item[(3)] The laws of physics are structurally stable with respect to perturbations whose amplitude is small relative to an appropriately defined metric which respects the primacy of $I_U$.  
\end{guess}

It is straightforward to define $I_U$ locally. There are two elements to the definition. The first, illustrated in Fig \ref{fractal}, shows in some three dimensional subset of state space, a single state-space trajectory segment (otherwise known as a `history') at some $j$th level of fractal iteration. At the $j+1$th iterate this trajectory is found to be a helical family of $N$ finer-scale trajectories winding around the coarse-scale trajectory. Zooming in further to the $j+2$nd fractal iterate (but not illustrated), each finer-scale trajectories would comprise a helix of $N$ yet finer-scale trajectory segments. This suggests $I_U$ comprises a fully fractal (i.e. indefinitely repeating self-similar) geometry and for simplicity one can imagine this to be. However, to emphasise here that invariant set theory can be considered a strictly finite theory, we assume that $I_U$ is actually a limit cycle and that the fractal iterations end at some finite $J \ggg 0$th iterate. Below we will use the word `fractal' to denote a measure-zero finite-$J$ fractal-like limit cycle. 
\begin{figure}
\centering
\includegraphics[scale=0.2]{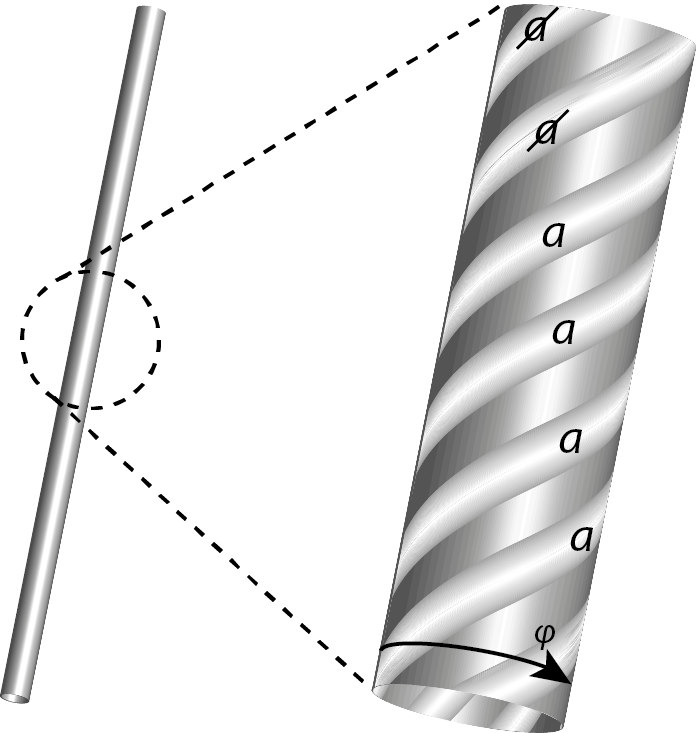}
\caption{\emph{A trajectory segment of $I_U$ at some `coarse' $j$th fractal iterate is found to comprise $N$ `fine-scale' $j+1$th iterate trajectory segments. These trajectory segments do not represent different universes, but are elements of a single trajectory of a mono-universe at different times (and different aeons) as it evolves on the compact set $I_U$. The $j+1$th iterate trajectory segments are wrapped helically around the $j$th iterate trajectory segments, and themselves comprise yet finer-scale $j+2$nd iterate trajectory segments of $I_U$, wrapped in a self-similar fashion. The $j+1$th iterate trajectory segments are each labelled symbolically $a$ or $\cancel a$ according to the regime to which they evolve under the process of decoherence and nonlinear clustering (illustrated in Fig \ref{decohere} below).}}
\label{fractal}
\end{figure}

Using the complex Hilbert vector structure introduced in Section \ref{complex}, the phase angle $\phi = 2\pi n/N$ provides a unique label for the $1 \le n \le N$th trajectory segment, implying an elemental or minimal angle $\delta \phi= 2\pi/N$ between one trajectory segment and its neighbour. In addition in Fig \ref{fractal} each $j+1$th trajectory segment is labelled `$a$' or `$\cancel a$' according to whether it evolves to one of two regimes or clusters - the second element of the definition of $I_U$ - as shown in Fig \ref{decohere}a. Here, the generic divergence of trajectories is associated with the process of decoherence as the system interacts with its environment - in essence it is the classical `butterfly effect'. At some finite amplitude, the trajectory segments begin to cluster into distinct regimes. Some fraction $n/N$ of these trajectories evolve to the `$a$' cluster, the remainder evolving to the `$\cancel a$' cluster. These regimes or clusters correspond to measurement eigenstates in quantum theory and, as discussed in Section \ref{gravity}, plausibly provide a primitive manifestation of the phenomenon of gravity. Since the environment involves a relatively large number of degrees of freedom, and in contradistinction with the helical evolution shown in Fig \ref{fractal}, Fig \ref{decohere}a is to be considered a two-dimensional projection of a process occurring in a relatively large dimensional subset of state space. The $j+2$th iterate trajectories associated with one of these $j+1$th trajectories is then shown in Fig \ref{decohere}a diverging to a second set of regimes labelled `$b$' and `$\cancel b$'. This would correspond to the process of sequential measurement in quantum physics. As shown in Fig \ref{decohere}b, from a fixed coarse-scale $j$th iterate perspective, it appears as if the trajectories are `branching', reminiscent of Everettian dynamics. However, such a perspective is illusory - there is no branching (moreover the $N$ trajectory segments do not correspond to `many worlds' but rather to the mono-universe $U$ at earlier or later epochs.)

\begin{figure}
\centering
\includegraphics[scale=0.3]{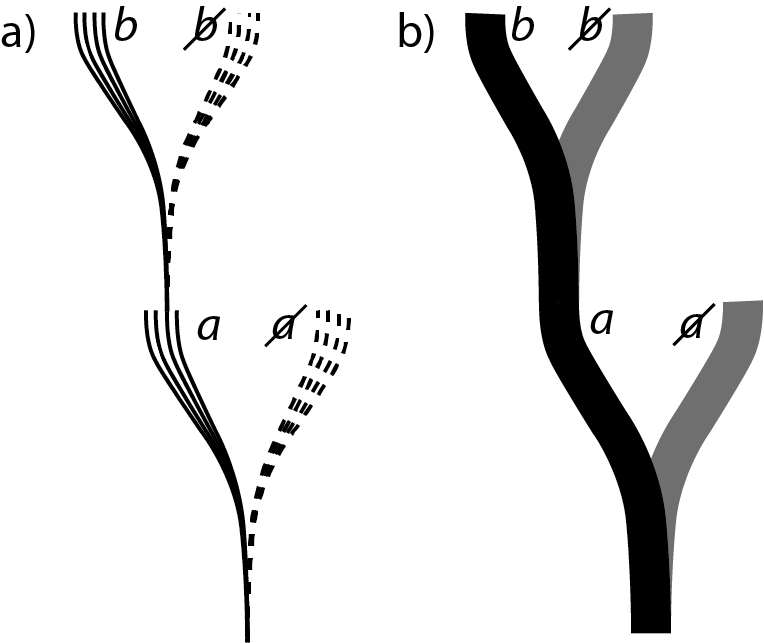}
\caption{\emph{a) $j+1$th iterate trajectories (`histories') on $I_U$ are shown initially diverging in state-space, reflecting the different ways in which a quantum system interacts with its environment,  and then clustering into two discrete regimes: $a$ and $\cancel a$. A $j+1$th iterate trajectory in one of the regimes itself comprises a family of $j+2$nd iterate trajectories evolving to one of a second pair of regimes: $b$ or $\cancel b$. In this way, time can be parametrised by fractal iterate number $j$ - see Section \ref{gravity}. b) From the perspective of the coarser $j$th iterate of $I_U$, this decoherent process resembles the type of `branching' that is often associated with the Everettian interpretation of quantum theory. From the perspective of the coarse-scale $j$th iterate, each branch has a different `weight' according to the corresponding number of $j+1$th trajectories it comprises.}}
\label{decohere}
\end{figure}

\subsection{The State-Space Metric $g_p$}
\label{gp}

Central to invariant set theory's interpretation of quantum physics are the gaps in the fractal structure illustrated in Fig \ref{fractal}; putative (typically counterfactual) states which lie in these gaps have no ontic significance.  Correspondingly, Hilbert vectors with values $\phi \ne 2 \pi n/N$, e.g. irrational values of $\phi$, cannot label uncertain trajectory segments on $I_U$. On the other hand, if $N$ is an inverse measure of the weakness of the gravitational force, which it is proposed to be, then we seek a theory in which $N$ is naturally large (but finite). However, as discussed, this immediately raises the problem of fine tuning: the larger is $N$ the smaller is the gap between the allowed rational values of $\phi$ and the smaller the apparent distance between points on $I_U$ and points off $I_U$. 

However, this conclusion above is implicitly dependent on the metric used to define distances in state space and hence on the notion of `closeness'. As humans from the earliest age, we learn to associate `closeness' in physical space with smallness of Euclidean distance, to the extent that this association is a deeply intuitive notion (and plausibly a key reason why theories of physics are conventionally based on $\mathbb R$ and $\mathbb C$ fields based on the completion of $\mathbb Q$ using the Euclidean metric. It may seem therefore seem reasonable to assume that distances in state space should also be measured using the Euclidean metric. However, the assumed primacy of fractal rather than Euclidean geometry in state space suggests that this may not be appropriate. As mentioned above, there are precisely two inequivalent classes of norm-induced metric on the rationals $\mathbb Q$: the Euclidean metric and the so-called $p$-adic metric, where for algebraic reasons $p$ is generally considered a prime number. By way of introduction to the $p$-adic metric, consider the sequence $\{1, 1.4, 1.41., 1.414, 1.4142, \ldots\}$ where each number is an increasingly accurate rational approximation to $\sqrt 2$. The sequence is Cauchy relative to the Euclidean metric $d_E(a,b)=|a-b|$, $a$, $b \in \mathbb{Q}$. Surprisingly perhaps, the sequence $\{1, 1+2, 1+2+2^2, 1+2+2^2+2^3, \ldots\}$ is also Cauchy, but  with respect to the ($p=2$) $p$-adic metric $d_p(a,b)=|a-b|_p$ where
\be
|x|_p=\left \{%
\begin{array} {ll}
p^{-\textrm{ord}_p x} &\textrm{if } x \ne 0 \\
0 &\textrm{if } x=0
\end{array}%
\right.
\ee
and
\be
\textrm{ord}_p x= \left \{%
\begin{array}{ll}
\textrm{the highest power of \emph p which divides \emph x, if } x \in \mathbb Z \\
\textrm{ord}_p a - \textrm{ord}_p b, \textrm{ if } x=a/b, \ \ a,b \in \mathbb Z, \ b \ne 0
\end{array}%
\right.
\ee
Hence, for example $d_2(1+2+2^2, 1+2)=1/4,\ \ d_2(1+2+2^2+2^3, 1+2+2^2)=1/8$. Just as $\mathbb{R}$ represents the completion of $\mathbb{Q}$ with respect to the Euclidean metric, so the $p$-adic numbers $\mathbb{Q}_p$ represent the completion of $\mathbb{Q}$ with respect to the $p$-adic metric. A general $p$-adic number can be written in the form
\be
\sum_{k=-m}^{\infty} a_k p^k
\ee
where $a_{-m} \ne 0$ and $a_k \in \{0,1,2, \ldots, p-1\}$. The so-called $p$-adic integers $\mathbb Z_p$ are those $p$-adic numbers where $m=0$.  They correspond to numbers $a \in \mathbb Q_p$ where $|a|_p \le 1$. Similarly, for two $p$-adic integers $a$ and $b$, we have $d_p(a,b) \le 1$, whilst if $c \in \mathbb Q_p$ but $c \notin \mathbb Z_p$, then $d_p(a,c) \ge p$ and $d_p(b,c) \ge p$ . 

As stated, the $p$-adic metric seems unintuitive. However, the $p$-adic integers have an important association with fractal geometry. Let $C(2)$ denote the familiar Cantor ternary set. Then the map $F_2: \mathbb Z_2 \rightarrow C(2)$ 
\be
\label{homeo}
F_2: \sum_{k=0}^{\infty} a_k2^k \mapsto \sum_{k=0}^{\infty} \frac{2a_k}{3^{k+1}} \textrm{ where } a_k \in \{0,1\}
\ee  
is a homeomorphism \cite{Robert}, implying that every point of $C(2)$ can be represented by a 2-adic integer. Under this homeomorphism, the 2-adic metric $d_2$ on $\mathbb Z$ is mapped to the Euclidean metric $d_E$ on $C(2)$. It is straightforward (see \cite{Katok} \cite{Robert}) to generalise this homeomorphism to a mapping $F_p$ between $\mathbb Z_p$ and a generalised Cantor set $\hat C(p)$ based on an iterated $(p-1)$-gon with an additional $p$th element at the centre: see Fig \ref{p5fractal} for an illustration of two iterates of $\hat C(5)$. With $p=N+1$, and given $N$ is divisible by 4, then if $p$ is prime it is a Pythagorean prime. In this case $\hat C(p)$ denotes the fractal set of $N$ trajectory segments of $I_U$ and the $N+1$th element at the centre of the $N$-gon. This can be considered to correspond to a single trajectory segment which evolves to the unstable equilibrium between clusters (and is thus labelled neither $a$ nor $\cancel a$). For large enough $N$, the probability that the true $I$th trajectory corresponds to this unstable equilibrium is negligible. From this perspective $I_U$ is locally a Cantor set of trajectory segments, i.e. $\hat C(p) \times \mathbb R$. 

\begin{figure}
\centering
\includegraphics[scale=0.3]{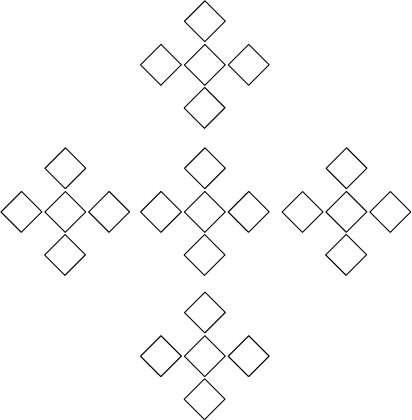}
\caption{\emph{Two fractal iterates associated with a self-similar Cantor Set $\hat C(p)$ associated with the smallest Pythagorean prime $p=5$. Locally, $I_U$ is equal to the set $\hat C(p)$ of trajectory segments, for large Pythagorean $p$. A metric $g_p$ is defined in the text so that a point $ z\notin \hat C(p)$ is necessarily distant from a point $x \in \hat C(p)$, no matter how close $x$ and $z$ are from a Euclidean perspective.}}
\label{p5fractal}
\end{figure}

Motivated by the properties of $d_p$, we define the following metric $g_p$ on cosmological state space. Let $x$, $y$ denote two trajectory segments on $I_U$, i.e. points on $\hat C(p)$, then $g_p(x,y)=d_E(x,y) \le 1$. From the fractal structure of $\hat C(p)$, the possible distances between $x$ and $y$ belong to the `quantised' set $\{1, 1/p, 1/p^2, 1/p^3, \ldots\}$. However, if $x \in \hat C(p)$ and $z \notin \hat C(p)$, then motivated by the fact that the distance between a $p$-adic integer and a $p$-adic non-integer number is at least $p$, we put $g_p(x,z)=p$. Since $p \gg 1$, then necessarily $g_p(x,z) \gg 1$, i.e. $z$ is $g_p$-distant from $x$ no matter how small is $d_E(x,z)>0$. Finally, if neither $x \ne y$ lie in $\hat C(p)$, then we again let $g_p(x, y)=p$. Otherwise, if $x=y$, $g_p(x,y)=0$. It is easily shown that $g_p$ satisfies the rules that define a metric (e.g. it satisfies the triangle inequality). 

There are two important consequences of using $g_p$ as the yardstick of distance in state space - the first technical, the second metaphysical. Firstly, it means that the algebraic properties of $p$-adic integers can be used to frame the laws of physics ($\mathbb Z_p$ is an integral domain). Here one can note that there is an extensive theory of $p$-adic Lie groups within which the Standard Model could conceivably be formulated in terms of the geometry of $I_U$. Secondly, if $z$ lies in a gap in $I_U$ and $x$ and $y$ are points on $I_U$ lying either side of the gap, then the $g_p$-distance between $x$ and $y$ is necessarily much smaller than between either $x$ and $z$, or $y$ and $z$. This means that the following assumption, made in Lewis's seminal treaty on causation \cite{Lewis}:
\begin{quote}
We may say that one world is closer to actuality than another if the first resembles out actual world than the second does.
\end{quote}
is not true (no matter how intuitively reasonable it may seem) if the first world does not lie on the invariant set and the second does. Similarly, if the hypothetical flap of a butterfly's wings takes a trajectory segment $x \in \hat C(p)$ to a trajectory segment $z \notin \hat C(p)$, then this flap cannot be considered a small perturbation, it is a necessarily a $g_p$-large-amplitude perturbation taking a state on $I_U$ to a state far away from $I_U$. This is clearly a non-classical concept - state space distances are taken to be Euclidean in classical theory and the flap of a butterfly's wings is a small perturbation. Hence invariant set theory is not classical. The consequences of these statements play a central role in this paper. Combining the Invariant Set Postulate and Definition (\ref{RI2}):
\begin{definition} 
\label{ISP2}
Let $\theta$, $\phi$ denote two angular degrees of freedom associated with a subset of $I_U$. With $\cos^2 \theta/2 = n_1/N$, $\phi/2\pi = n_2/N$, $0 \le n_1 \le N$, $0 \le n_2 \le N$, then the complex Hilbert vector 
\be
|\psi(\theta, \phi)\rangle=\cos \frac{\theta}{2} \; |a\rangle+  e^{i \phi} \sin \frac{\theta}{2} \;|\cancel a \rangle
\ee
is a representation of some fixed but unknown trajectory segment of $I_U$ from an ensemble of $N \gg 0 $ $j+1$th iterate state-space trajectory segments on $I_U$. The $N$ trajectories can be labelled by the distinct state-space regimes ($a$ and $\cancel a$) to which they evolve. The fraction $n_1/N$ of $j+1$th trajectory segments labelled `$a$' is equal to $\cos^2 \theta/2$, which therefore provides a `weighting' for the corresponding $j$th iterate trajectory. The $n_2$th $j+1$th iterate trajectory can labelled by $\phi=2\pi n_2/ N$. Hence, necessarily, $\cos \theta \in \mathbb Q$ and $\phi/2\pi \in \mathbb Q$. 
\end{definition}

We can embed such realistic complex Hilbert vectors into the complex one dimensional Hilbert Space $\mathscr H$. However, the corresponding Hilbert vectors not associated with rational $\cos \theta$ and $\phi/2\pi$ have no interpretation as uncertain trajectories on $I_U$. Using $g_p$ as the preferred yardstick in state space, such Hilbert vectors are distant from realistic vectors based on rational $\cos \theta$ and $\phi/2\pi$. Such vectors can be imagined to represent hypothetical uncertain counterfactual trajectories which lie off $I_U$ and are therefore far from $I_U$. By virtue of the Invariant Set Postulate, such hypothetical trajectories have a different ontological status to those that lie on $I_U$: they do not arise from the laws of physics and are therefore physically unreal. Since the clusters $a$ and $\cancel a$ are mutually exclusive, the Hilbert states $|a\rangle$ and $|\cancel a\rangle$ will be orthogonal. Since orthogonal vectors are generated as eigenvectors of Hermitian matrices, and since invariant set theory is essentially a finite theory, we will borrow from quantum theory the notion that particular pairs of vectors $|a\rangle$ and $|\cancel a\rangle$ (associated with properties of position, momentum, angular momentum etc) are eigenvectors of the relevant Hermitian matrices.  The discussion above generalises straightforwardly when the corresponding Hermitian matrix has multiple eigenvectors. 

In general, it is not possible to compute whether two points in state space are $g_p$-close or not. In the case where $I_U$ is a true fractal, then its geometric properties are generically non-algorithmic \cite{Blum}. In the case where $I_U$ is a $J$-finite limit cycle, then $I_U$'s geometric properties, whilst algorithmic, are computationally irreducible \cite{Wolfram} implying that they cannot be determined from a subset of $I_U$, e.g. by a computer, no matter how big. This will be relevant in resolving the conceptual problems associated with delayed-choice (and related) experiments in quantum physics. 

\section{The Elements of Quantum Physics}

\label{Schrodinger}

It is now possible to challenge the assumption that $\mathbb R$ and $\mathbb C$ (and associated algebraic/geometric structures) define the arenas within which physical theory, both classical and quantum, is formulated. The discussion starts with the most primitive notion in quantum physics - complementarity - and it is shown how this arises from Theorem \ref{theorem}. As will be seen, the Schr\"{o}dinger (and Dirac) equations emerge as a singular limit \cite{Berry} of evolution on a particular $I_U$, when the fractal parameter $N$ is set equal to infinity. Throughout the role of $g_p$ is critical and encapsulates the fact that invariant set theory is not classical, even though locally realistic. 

\subsection{Quantum Complementarity, EPR and Delayed Choice}
\label{comp}

In invariant set theory, the quantum notion of complementarity arises from a number-theoretic inconsistency between between the two Hilbert vectors
\begin{align}
\label{onetotwo}
|\psi_1\rangle &= \cos\frac{\phi}{2}|a\rangle -i \sin\frac{\phi}{2} |\cancel a\rangle \nonumber \\
|\psi_2\rangle &= \ \ \frac{1}{\sqrt 2}(|a\rangle + e^{i \phi} |\cancel a\rangle) 
\end{align}
related by the Hadamard transform (\ref{Hadamard}). Consider two experimental configurations of a Mach-Zehnder interferometer (Fig \ref{mach}) where the lengths of the arms of the interferometer are for all practical purposes equal, i.e. the phase angle $\phi$ associated with the difference in the length of the arms is equal to zero to experimental accuracy. Hence, if on Monday an experimenter performs a momentum measurement on some particular particle (Fig \ref{mach}a), it is virtually certain that the detector $D_a$ will be triggered, and not the detector $D_{\cancel a}$. However, if on Tuesday she performs a position measurement on the some second particle (Fig \ref{mach}b), there is an equal chance that either detector will be triggered. 

\begin{figure}
\centering
\includegraphics[scale=0.3]{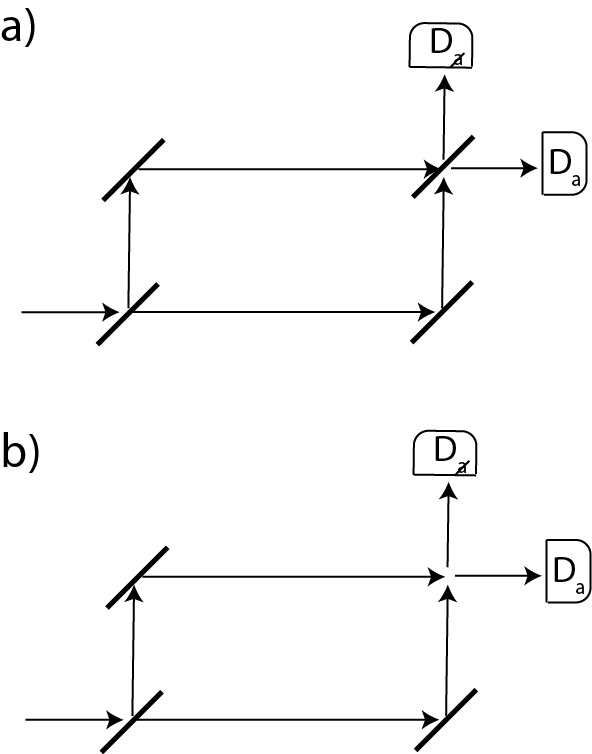}
\caption{\emph{a) On Monday an experimenter performs a momentum measurement using a Mach-Zehnder interferometer. b) On Tuesday she performs a which way (position) measurement removing the second half-silvered mirror. The lengths of the two arms of the interferometer are, to experimental accuracy, equal. We consider the questions: what would the experimenter have measured had she measured the position of Monday's particle and the momentum of Tuesday's particle.}}
\label{mach}
\end{figure}

According to invariant set theory, whatever the exact value $\phi_{\text{Mon}}$ associated with Monday's experiment, $\cos\phi_{\text{Mon}}$ must necessarily be rational. By contrast, whatever the precise value $\phi_{\text{Tues}}$ associated with Tuesday's experiment, $\phi_{\text{Tues}}/2\pi$ must be rational. One simple way to satisfy these two constraints is to assert that $\phi_{\text{Mon}}=\phi_{\text{Tues}}=0^{\circ}$ precisely. However, this is impossible. There are many external affects (the movement of heavy trucks outside the lab and so on) which prevent $\phi$ from being set to some chosen value \emph{precisely}. With care, Alice can shield her apparatus from many of these external effects. However, as a matter of principle she cannot shield her apparatus from the ubiquitous gravitational waves arising e.g. from distant astrophysical events. That is to say, gravity is the one effect that, as a matter of principle, prevents one part of the universe being considered truly isolated from the rest and provides a fundamental reason why $\phi_{\text{Mon}}=\phi_{\text{Tues}}=0^{\circ}$ is impossible. 

Hence, according to Theorem \ref{theorem} we can assert that if the experimenter measured the momentum of a particle on Monday, she would not have been able to infer a well-defined outcome for a (which-way) position measurement of the same particle on Monday. Similarly, she would not have been able to infer a well-defined outcome for a momentum measurement on Tuesday's particle. 

From an experimenter's point of view this result may seem inconsistent with the fact that Tuesday's angle can be made arbitrarily close to Monday's angle by making the experimental apparatus sufficiently precise. However, when considering the implications of quantum no-go theorems below, counterfactual questions of a certain and characteristic type become relevant: What would the experimenter have observed had she measured the position of some particular particle on Monday having actually measured momentum, and/or the momentum of some particular particle on Tuesday, having actually position. The key question, then, is whether the difference between  $\phi_{\text{Mon}} \approx  \phi_{\text{Tues}}$ on the one hand, and $\phi_{\text{Mon}} \approx  \phi_{\text{Tues}}$ on the other, actually matters? More specifically, is there a plausible framework for physical theory where the latter is the singular \cite{Berry} rather than the smooth limit as experimental accuracy/precision goes to infinity. Invariant set theory, and the corresponding metric $g_p$ provides such a singular limiting framework. As such, the fact that the experimenter can make Tuesday's angle arbitrarily close to Monday's angle with respect to the Euclidean metric of space-time, is irrelevant. 

In quantum theory, by the non-commutation of position and momentum operators in Hilbert Space, these counterfactual states cannot be interpreted realistically: a particle with well-defined momentum has no well-defined position and \emph{vice versa}. In invariant set theory, these same conclusions are reached for quite different (number-theoretic) reasons. That is to say, the Hadamard transform $U_H$ applied to $|\psi_1\rangle$ leads to a Hilbert vector which has no realistic interpretation because $e^{i \phi_{\text{Mon}}}$ is undefined for angles $\phi$ where $\cos \phi_{\text{Mon}}$ is rational. In invariant set theory, this `unrealistic' Hilbert state must therefore correspond to a hypothetical uncertain trajectory segment lying in a gap off $I_U$, associated with a putative counterfactual world $U'$ where a position measurement was made on Monday. Although $U'$ can be imagined, because $U' \notin I_U$, it is not an element of physical reality. Using the metric $g_p$ as the yardstick of distance in state space, $U'$ is distant from the real world $U \in I_U$ where momentum was measured and therefore $U'$ cannot be reached by small-amplitude perturbations of states on $I_U$. 

How then can we interpret the Hadamard transformation as a mapping which takes Monday's state into Tuesday's state? Since $\cos\phi_{\text{Mon}}$ and $\phi_{\text{Tues}}/2\pi$ are rational, and $\phi_{\text{Mon}} \ne \phi_{\text{Tues}}$, then, in invariant set theory, the transformation which maps $|\psi_1\rangle$ to $|\psi_2\rangle$ cannot be of the form (\ref{Hadamard}), but must instead be of the form
\be
\label{Hadamardprime}
\bar U_H=\frac{1}{\sqrt 2} \bp 1 &1 \\ 1& -1 \ep  \ \bp \; 1 &0 \\ 0& e^{i \delta \phi} \; \ep
\ee
where $\delta \phi = \phi_{\text{Tues}}-\phi_{\text{Mon}} \ne 0$ precisely, but is $\approx 0$ to experimental accuracy. Of course it does not require any conscious effort on the part of the experimenter to realise such an $\bar U_H$ - it is implicit in the laws of dynamical evolution which ensure that states remain on $I_U$ and hence that $\delta \phi$ is such as to map a $\phi_{\text{Mon}}$ with $\cos \phi_{\text{Mon}} \in \mathbb Q$ to a $\phi_{\text{Tue}}$ with $\phi_{\text{Tue}} \in \mathbb Q$. The behaviour of the system is singular and not smooth in the limit $\delta \phi =0$. Consistent with this, a universe where $\delta \phi=0$ is necessarily $g_p$ distant from a universe $\in I_U$, no matter how small (in the Euclidean sense) is $\delta \phi$. 

To aid understanding of this result, let us imagine a variant of invariant set theory where instead of (\ref{onetotwo}), we have
\begin{align}
|\psi_1\rangle &= \sqrt{F(\phi/2\pi)}|a\rangle -i \sqrt{1-F(\phi/2\pi)} |\cancel a\rangle \nonumber \\
|\psi_2\rangle &= \ \ \frac{1}{\sqrt 2}(|a\rangle + e^{i \phi} |\cancel a\rangle) \nonumber
\end{align} 
where $F$ is a polynomial with rational coefficients, which approximates $\cos^2 \phi/2$ to Euclidean (and hence experimental) accuracy. Hence if $\phi$ is a rational multiple of $2\pi$ then, because of the ring-theoretic properties of polynomials, so is $F(\phi/2\pi)$. In this circumstance, there would be no number theoretic incommensurateness between $|\psi_1\rangle$ and $|\psi_2\rangle$ of the type discussed above. This would lead to the most blatant inconsistency with the laws of probability, of the type described e.g. in \cite{FeynmanHibbs}: if we take a realistic perspective where the particle either travels through the upper or lower arm of the interferometer, then the very notion of probability demands that the probability $P_a$ of being detected by $D_a$ in Fig \ref{mach}a must be the sum $P_1+P_2$ of the probabilities of the particle travelling through the upper arm of the interferometer and the lower arm, respectively, as found from Fig \ref{mach}b. Manifestly, this is not the case. Hence, it we insist that physics is realistic, the squared amplitudes cannot be algebraic functions of angle. As before, such a statement would make no sense if the yardstick of distance in state space were the standard Euclidean metric since the Euclidean distance between polynomial and transcendental functions can be arbitrarily small. 

The analysis above reveals an attractive property of invariant set theory which conventional physical theory (including quantum theory) based on $\mathbb R$ and $\mathbb Q$ does not share. Invariant set theory is fundamentally a finite theory and inconsistencies arise when putative states are considered which violate this finiteness (i.e. where descriptors take irrational rather than rational values). The crucial role of the transcendental nature of the cosine function does not contradict the primacy of finiteness: the transcendence of the cosine function merely guarantees that rational angles and rational cosines of angles are generically incompatible (i.e. no matter how large is $N$) and this in turn ensures that in an entirely realistic framework there can be no inconsistency in describing the quantum state probabilistically, in the usual frequentist sense.

This leads to a new perspective on EPR. Suppose Alice and Bob choose to measure either the momentum or the position of their halves of an entangled particle pair. By the discussion above, if Alice were to choose to measure momentum she could not have measured position, and \emph{vice versa}. Similarly for Bob. Now if Alice measures momentum and Bob position, then conservation of momentum will allow Alice to infer the position of her particle and Bob the momentum of his particle. In this situation, there is no inconsistency in concluding that these particles have well-defined position and momentum, since only one of these properties can be measured directly (the other is inferred from a measurement on the other particle). If both Alice and Bob measure momentum then their measurements will be found to be consistent with conservation of momentum, and neither could have measured position. It does not follow in this latter case that neither particle `has' a position. Instead one must conclude that the positions of these particles (at the time the momentum measurements are being made) play no direct role in the way in which the particles interact with the rest of the universe. That is to say, the positions play no role in determining the geometry of $I_U$. More generally, the number-theoretic incommensurateness discussed above neither contradicts nor demands the realist notion that sub-systems have a full and definite set of properties - though common sense suggests that it is reasonable to assume such sets of properties do exist. It might be thought that such a conclusion is inconsistent with the Kochen-Specker theorem. However, it is known that this theorem (like the ones discussed here) can be nullified by finite-precision experimentation \cite{Meyer:1999}.

Let us now consider the conceptual problem of a delayed-choice experiment, i.e. where the decision on whether to remove the second beam splitter in Fig \ref{mach} is made at a time $t_1>t_0$, where $t_0$ denotes the time the particle has passed through the first beam splitter. It is certainly the case that at $t_0$ the cosine of the phase angle difference $\phi$ must be rational if the experimenter decides at $t_1$ to make a momentum measurement. Is this inconsistent with experimenter free will? Certainly not. As discussed in Section \ref{padic}, the properties of $I_U$ are not computational (either non-algorithmic or computationally irreducible). This means that it is impossible to define a space-time event at $t_0$, such as the ringing of a bell, which can signify an outcome either of an experiment or a physical computation that establishes whether $\cos \phi \in \mathbb Q$ or not. Hence there exists no space-time event at $t_0$ that can cause the decision at $t_1$. Hence there is no violation of experimenter free will. Similarly, there exists no space-time event at $t_0$ to be retrocaused by the experimenter's decision at $t_1$ \cite{Price:1997}. As discussed in \cite{Palmer:2015}, there are parallels with the global geometry of black-hole event horizons in general relativity. All of the problems with causality in quantum physics can be resolved by noting that the metric $g_p$ of state space is either not computational (non-algorithmic or computationally irreducible). The role of non-computability in quantum physics has already been speculated about by Penrose \cite{Penrose:1989}.

\subsection{Spin and the Sequential Stern-Gerlach Experiment}
\label{SG}

Given that the state-space clusters $a$ and $\cancel a$ are mutually exclusive, it is natural to associate, as in quantum theory, the base vectors $|a\rangle$ and $|\cancel a\rangle$ with non-degenerate eigenvectors of some appropriate Hermitian operator (since such eigenvectors are orthogonal). Here we consider clusters based on eigenvectors of the Pauli spin operators $\{\sigma_x, \sigma_y, \sigma_z\}$ and show that a very similar number-theoretic incommensurateness to that discussed in Section \ref{comp} prevent generic simultaneous measurements of spin. 

Let $A$, $B$ and $C$ denote three arbitrary points on the sphere with centre $O$ representing directions $\hat{\mathbf a}$, $\hat{\mathbf b}$ and $\hat{\mathbf c}$ in physical space, respectively (see Fig \ref{spheret}). Let the $\hat{\mathbf{b}}$ direction correspond to the $z$-direction, so that the spin matrix corresponding to the $\hat{\mathbf b}$ direction is $S_b=\sigma_z$
with eigenvectors
\be
|b\rangle=
\begin{pmatrix}
1\\
0
\end{pmatrix}
\ \ \text{and}\ \ 
|\cancel b\rangle =
\begin{pmatrix}
0\\
1
\end{pmatrix}
\ee
Let the great circle between $A$ and $B$ lie in the $Oxz$ plane so that the spin matrix corresponding to the $\hat{\mathbf a}$ direction is given by $S_A= \sigma_z \ \cos \theta_{AB}+ \sigma_x\ \sin \theta_{AB}$ with eigenvectors
\be
|a \rangle=
\begin{pmatrix}
\cos \frac{\theta_{AB}}{2}\\
\sin \frac{\theta_{AB}}{2}
\end{pmatrix}
\ \ \text{and}\ \ |\cancel a\rangle =
\begin{pmatrix}
 \sin \frac{\theta_{AB}}{2}\\
-\cos \frac{\theta_{AB}}{2}
\end{pmatrix}
\ee
so that
\be
|b\rangle
= \cos \frac{\theta_{AB}}{2} |a\rangle + \sin \frac{\theta_{AB}}{2} |\cancel a\rangle
\ee
Finally, by spherical geometry
\be
S_C= \sigma_z \ \cos \theta_{BC} +(\sigma_x \cos \gamma + \sigma_y \sin \gamma) \ \sin\theta_{BC}
= 
\begin{pmatrix}
\cos\theta_{BC}& e^{-i \gamma}\sin \theta_{BC} \\
e^{i \gamma}\sin \theta_{BC}& -\cos\theta_{BC}
\end{pmatrix}
\ee
with eigenvectors
\be
|c\rangle=
\begin{pmatrix}
 \cos \frac{\theta_{BC}}{2}\\
e^{i \gamma} \sin \frac{\theta_{BC}}{2}
\end{pmatrix}
\ \ \text{and}\ \ 
|\cancel c\rangle =
\begin{pmatrix}
e^{-i \gamma} \sin \frac{\theta_{BC}}{2}\\
-\cos \frac{\theta_{BC}}{2}
\end{pmatrix}
\ee
so that
\be
|b\rangle
= \cos \frac{\theta_{BC}}{2} |c\rangle + e^{i \gamma} \sin \frac{\theta_{BC}}{2} |\cancel c\rangle
\ee
Suppose it is the case that: a) $\cos \theta_{AB} \in \mathbb Q$, b) $\cos \theta_{BC} \in \mathbb Q$ and c) $\gamma/ 2\pi \in \mathbb Q$. 
\begin{figure}
\centering
\includegraphics[scale=0.3]{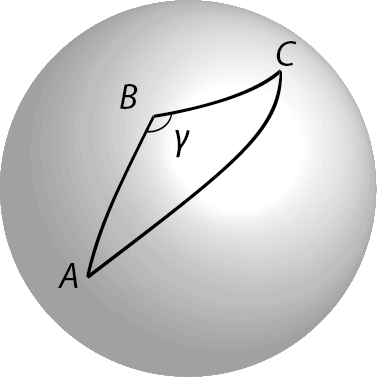}
\caption{\emph{Number-theoretic properties of the spherical triangle $\triangle ABC$ play a central role in interpreting quantum `weirdness' realistically, in invariant set theory}}
\label{spheret}
\end{figure}
Now consider a crucial question. Is it possible for $\cos \theta_{AC} \in \mathbb Q$? By  the cosine rule for spherical triangles
\be
\cos \theta_{AC}= \cos \theta_{AB}\cos \theta_{BC}+\sin \theta_{AB}\sin \theta_{BC}\cos \gamma
\ee
The right hand side is the sum of two terms. The first is rational since it is the product of two terms each of which, by construction, is rational. The second is the product of three terms the last of which, $\cos \gamma$,  is irrational, except for the eight exceptions listed above. Let us assume $A$, $B$ and $C$ lie approximately on a great circle, ie $\gamma \approx 180 ^\circ$ to experimental accuracy. By the discussion above (ubiquity of unshieldable gravitational waves), it cannot be that $\gamma=180 ^\circ$ precisely. Hence $\cos \gamma$ is irrational. Since $\theta_{AB}$, $\theta_{BC}$ and $\gamma$ are independent degrees of freedom defining the triangle $\triangle ABC$, there is no reason why $\sin \theta_{AB}$ and $\sin \theta_{BC}$ should conspire with $\cos \gamma$ to make the product $\sin \theta_{AB}\sin \theta_{BC}\cos \gamma$ rational. Hence $\cos \theta_{AC}$ is the sum of a rational and an irrational and is therefore irrational. 

As an illustration of the physical relevance of this result, consider an ensemble of spin-1/2 particles prepared by the first of three Stern-Gerlach apparatuses with spins oriented in the direction $\hat{\mathbf {a}}$ in physical 3-space. The particles pass through a second apparatus oriented in the direction $\hat{\mathbf{b}}$. The particles that are output along the spin-up channel of the second apparatus are then passed into a third Stern-Gerlach apparatus oriented in the direction $\hat{\mathbf{c}}$. As above we consider the corresponding points $A$, $B$ and $C$ to lie approximately on a single great circle, consistent with coplanar measurement orientations. However, precise coplanarity is impossible to achieve \emph{precisely} because of the ubiquity of space-time ripples, as in the discussions above. Hence, as in Fig \ref{spheret} we assume that $A$, $B$ and $C$ are the vertices of some non-degenerate triangle $\triangle ABC$. 

According to invariant set theory, all of $\cos \theta_{AB}$, $\cos \theta_{BC}$ and $\cos \gamma$ must be rational for the sequential Stern-Gerlach experiment to lie on $I_U$. If this is the case, then, as discussed above, $\cos \theta_{AC}$ cannot be rational. This means that the counterfactual experiment where the order of the second and third Stern-Gerlach apparatuses is reversed from ABC to ACB cannot lie on $I_U$. Of course in quantum theory, this result is expressed through the non-commutativity of the spin operators. In (the finite) invariant set theory, it arises, again, through number-theoretic incommensurateness associated with the transcendental nature of the cosine function. 

We can of course envisage performing two separate sequential Stern-Gerlach experiments (one on a Monday, the other on a Tuesday, say) where the order of the Stern-Gerlach apparatuses was $ABC$ and $ACB$ respectively. For Monday's experiment, $\cos \theta_{AB}$ and $\cos \theta_{BC}$ are rational, and the angle subtended at $B$ is a rational multiple of $2 \pi$. For Tuesday's experiment, $\cos \theta_{AC}$ and $\cos \theta_{BC}$ are rational, and the angle subtended at $C$ is a rational multiple of $2 \pi$. As before, this would be impossible if the triangle $\triangle ABC$ was \emph{precisely} the same on Monday and Tuesday. However, this will not be the case - background space-time ripples are necessarily different on Tuesday compared with Monday. That is to say, if Monday's triangle is $\triangle ABC$, and Tuesday's triangle is $\triangle A'C'B'$ and $A \approx A'$, $B \approx B'$ and $C \approx C'$ to (any nonzero) experimental accuracy, then the limit $A = A'$, $B = B'$ and $C=C'$ is singular and $g_p$ distant from the experiments on $I_U$.  (The limit can be smooth when any two primed and unprimed points are set equal, but is singular when all three are set equal, reminiscent of Penrose's `impossible' triangle, which is only impossible because of an inappropriate assumption of $2D$ - rather than $3D$ - Euclidean metric.) An almost identical argument will be used to show in Section \ref{CHSH} that no experiment has ever demonstrated that the Bell Inequalities are violated, even approximately!

\subsection{The Schr\"{o}dinger and Dirac Equations}

Consider an isolated particle of mass $m$ moving uniformly along the $x$ axis in physical space and a particle detector located at $x_D$. Following the discussion in that section the helical structure of fine-scale trajectories illustrated in Fig \ref{fractal} are represented by the bit string
\be
\label{sevolution1}
S(t) = \zeta^n S  \equiv e^{-i \omega t} S
\ee
where $t$ parametrises length along a trajectory, $n/N=\omega t / 2\pi$, the exponential represents a bit-string operator and $\omega$ denotes the frequency of the fine-scale trajectories relative to time along the coarse-scale trajectory. As discussed above, the operator $e^{-i\omega t}$ requires $\omega t$ to be a rational multiple of $2\pi$. 

This can be generalised to include variations in the position $x_D$ of the detector relative to the particle. 
\be
\label{sevolution2}
S(x,t) = \zeta^m S  \equiv e^{-i kx_D} S(t)
\ee
Using the Invariant Set Postulate, we can infer:
\newtheorem* {ever}{De Broglie Relations}
\begin{ever}
\label{EP}
\item The energy-momentum $(E,p)$ of the isolated particle is determined by the geometry of $I_U$ and in particular by the temporal frequency $\omega$ of the $j+1$th iterate trajectory segments as they wrap around the corresponding $j$th iterate trajectory segment, and the wavenumber $k$ of phase variations associated with different positions $x_D$ of the detector relative to the particle. That is 
\be
\label{db}
E=\hbar \omega;\ \ p=\hbar k
\ee
implying that $\hbar$ is the constant of nature that links energy-momentum (the source of space-time geometry) to the geometry of $I_U$. 
\end{ever}
Substituting (\ref{db}) into the non-relativistic formula $E=p^2/2m$, consistent with the use of classical descriptions on the coarse-scale trajectories, we have 
\be
\label{Schrodinger1}
\hbar \omega - \frac{\hbar^2 k^2}{2m}=0
\ee
which, in the singular limit $N=\infty$ and dropping the subscript `$D$' can be written in the familiar differential-equation form of the Schr\"{o}dinger equation:
\be
\label{Schrodinger2}
\left(i \frac{\partial }{\partial t} + \frac{\hbar}{2m} \frac{\partial^2 }{\partial x^2}\right) |\psi\rangle= 0
\ee
with integral solution
\be
\label{intsol}
|\psi(x,t)\rangle= e^{-iEt/\hbar} \;|\psi(x,0)\rangle;\ \ \ \
|\psi(x,0)\rangle= e^{ipx/\hbar} \ |\psi(0,0)\rangle
\ee
and where the complex exponentials are the familiar units of the algebraically closed field $\mathbb C$. The finiteness of $N$ implies that space-time itself must be finite and hence granular. 

It is straightforward to generalise the bit-string construction to describe a particle in a potential well $V(x)$. This is achieved by `continuing' the exponential permuation/negation operator so that for $x\ge 0$, 
\be
e^{-x} S = \{\; \underbrace{a \; a \ldots a}_{\frac{N}{2} (1+e^{-x})} \ \underbrace{\cancel a\;  \cancel a \ldots \cancel a}_{\frac{N}{2} (1-e^{-x})} \}= \sqrt{\frac{1+e^{-x}}{2}}|a\rangle + \sqrt{\frac{1-e^{-x}}{2}}|\cancel a\rangle 
\ee
where the first exponential is a bit-string operator, and is the usual exponential function elsewhere. As before, this requires $e^{x}$ to be rational which means $x \ne 0$ cannot itself be rational \cite{Niven}. This leads to the familiar quantum tunnelling effect. 

Since a key objective is the development of a theory of quantum physics which is compatible with general relativity, invariant set theory must certainly be shown to be compatible with special relativity. 
\begin{figure}
\begin{center}
\includegraphics[scale=0.2]{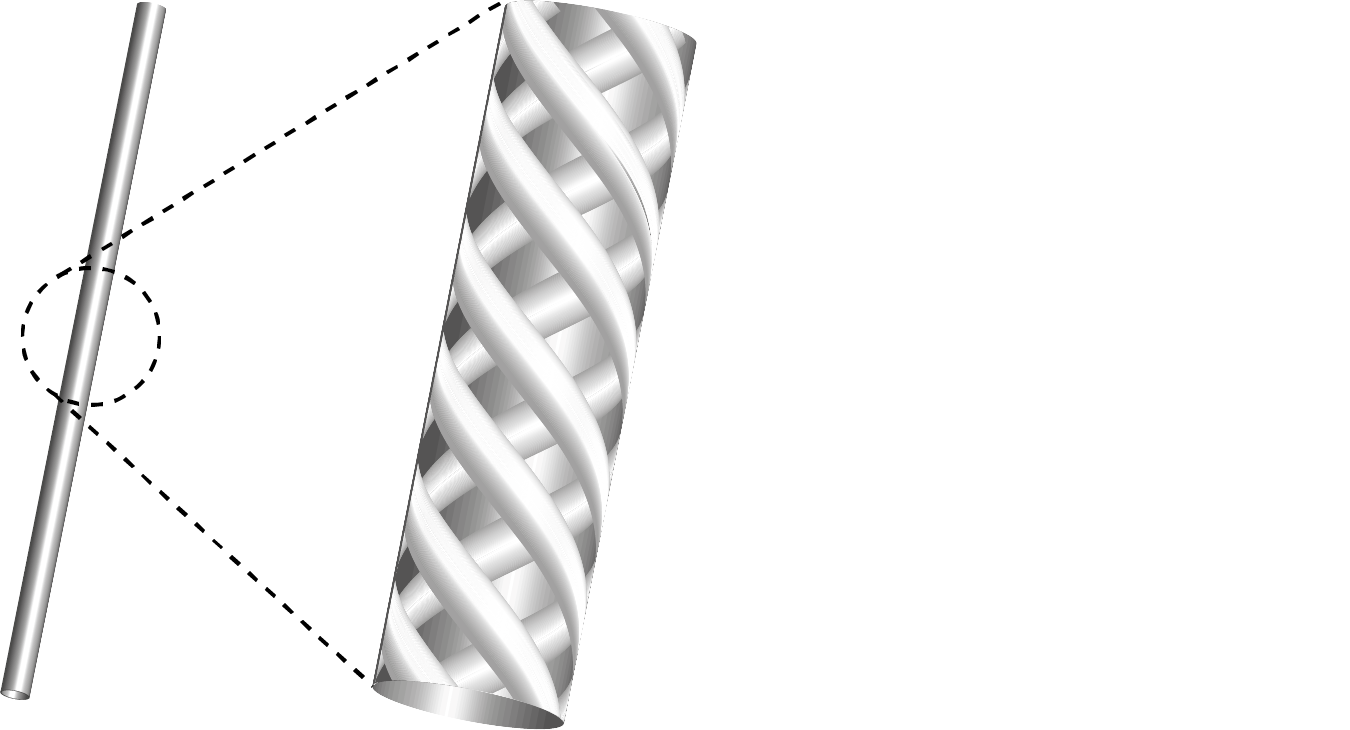}
\end{center}
\caption{\emph{Fractal trajectory structure that incorporates the relativistic positive/negative frequency splitting of Dirac theory.}}
\label{fracdirac}
\end{figure}
As is well known, relativistic invariance requires us to consider positive and negative frequencies together. This is achieved by extending the fine-scale helix combining $N$ clockwise helical trajectories with $N$ anticlockwise helical trajectories - a double helix (see Fig \ref{fracdirac}). This can be represented as 
\be
\label{evolution1}
S_{*}(t) = \zeta^n S \| \zeta^{-n}S. 
\ee
where $\|$ denotes the concatenation operator and $S_*$ is a $2N$-component bit string, or alternatively as 
\be
\label{evolution4}
S_{*}(t)=
\begin{pmatrix}
\zeta^{n} &0 \\
0& \zeta^{-n}
\end{pmatrix}
S_{*}(0)=
\begin{pmatrix}
e^{-i \omega t} &0 \\
0& e^{i \omega t}
\end{pmatrix}
S_{*}(0)
\ee
Again, only in the \emph{singular} limit at $N=\infty$, can $e^{i \omega t}$ be extended over all complex phases. That is to say, at the singular limit $N=\infty$, $e^{i \omega t}$ can be identified with the familiar complex exponential function and $S_{*}(t)$ considered a solution to the equation
\be
\label{evolution3}
i\gamma_0\; \partial_t \psi - \omega \psi = 0
\ee 
where $\gamma_0$ is the familiar Dirac matrix. With $E=\hbar \omega$ based on Postulate \ref{ISP} and $E=mc^2$ from special relativity, then (\ref{evolution3}) corresponds to the integral solution of the zero-momentum form of the Dirac equation,
\be
\label{Dirac1}
i\hbar \gamma_0\; \partial_t \psi - mc^2 \psi = 0
\ee
To keep this paper to a manageable size, a description of the particle in motion is deferred to another paper. To do this, $I_U$ must be extended to include the quaternionic geometry of $I_U$. As discussed in Section \ref{complex}, (see (\ref{quaternions}), quaternionic negation/permutation operators are readily defined. As is well known, the Lorentz group is locally isomorphic to the quaternionic group. Based on the assumption that the geometry of $I_U$ is primitive, then it would therefore appear that the causal Lorentzian structure of space-time may itself be considered emergent from the fractal geometry of $I_U$.  

This analysis suggests an interpretation of the Schr\"{o}dinger equation as a powerful computational tool, but one whose underpinning Euclidean function spaces (at $p=N=\infty$) are too smooth to be able to distinguish ontological and non-ontological states. An ability to distinguish ontological and non-ontological states becomes crucial in evading so-called quantum no-go theorems below. 

\subsection{Measurement}
\label{measure}

As discussed above, the process of measurement in invariant set theory is described firstly with a linear instability associated with a generic interaction with between the system and its environment, and a subsequent nonlinear clustering into distinct regimes in state space. As discussed above, in the Hilbert state formalism these distinct clusters are represented by the orthogonal eigenvectors of a Hermitian operator. Throughout this paper attention is focussed on cases comprising just two clusters, labelled $a$ and $\cancel a$, but the extension to multiple clusters is straightforward. The unpredictability of quantum measurement arises simply because, at any fractal iterate, the value $I$, of the $I$th trajectory segment corresponding to reality, is uncertain - it can correspond to any of the $N$ trajectories with equal probability. As such this uncertainty is merely a manifestation of the fact we (humans) have no way of knowing which aeon of cosmic evolution we are currently experiencing. `God' knows this, of course, and therefore has no recourse to dice! Put another way, there need be nothing intrinsically stochastic or indeterministic about the laws of physics. 

Consider a quantum system that undergoes repeated measurement. The measurement outcomes define a bit string such as $O=\{a,a,\cancel a, \cancel a, a \ldots\}$ - a string of symbolic labelings of the relevant clusters at different fractal iterates $j$ (e.g. extend the clustering in Fig \ref{decohere}a indefinitely upwards). Hence $O$ defines a measurement frequency in space-time. The string $O$ is conceptually different to the strings defined in Section \ref{interfere} which correspond to neighbouring trajectory segments on $I_U$ from which probability is defined.  The correspondence between probability and frequency is theoretically problematic in conventional physics based on $\mathbb R$ and $\mathbb C$ \cite{Wallace}. However, these concepts are easily related when the geometry underpinning state space is fractal. In particular, the outcomes of repeated measurements can be represented by a point $X \in \hat C(p)$ and hence by the $p$-adic integer $...x_3x_2x_1.$ where $x_i \in \{0,1,2, \ldots p-1\}$. Here $x_i$ defines the element of the $i$th iterate of the Cantor Set $\hat C(p)$ in which $X$ lies. $\hat C(p)$ comes with a natural measure: the Haar Measure. With respect to this measure, the probability that $x_i$ equals any of the digits in $\{0,1,2, \ldots p-1\}$ is equal to $1/p$. The following theorem equates frequency of occurrence to probability: 
\bigskip

$\mathbf{Theorem}$ \cite{Ruban} Let $X$ be a typical element of a Cantor set $\hat C(p)$, i.e. an element drawn randomly with respect to the Haar measure. Then with probability one, the frequency of occurrence of any of the digits $\{0,1,2, \ldots p-1\}$ in the expansion for $X$ is equal to $1/p$. 
\bigskip

Explicit formulae for the divergence and clustering of trajectories in state space are not given in this paper. To do so would require a more explicit and detailed analysis of the global (rather than local) geometry of $I_U$. Nevertheless, some general comments can be made on similarities and differences with existing approaches to the measurement problem:

\begin{itemize}

\item As discussed above, invariant set theory does not presume any fundamental  significance to the superposed Hilbert vector - it is a probabilistic representation of some realistic but uncertain state of the system. Hence, in invariant set theory, like Bohmian theory, no `reduction'  - or `collapse of the wavefunction' - is actually taking place during the measurement (i.e. clustering) procedure. 

\item Indeed one can interpret the quantum potential of Bohmian theory as a coarse-grained Euclidean $L^2$ representation of the fractal invariant set geometry in state space. This coarse-grained Euclidean representation smears over the all-important fractal gaps in $I_U$. This means that the Bohmian quantum potential does not have the fine-grained ontological properties of the invariant set and therefore, as discussed in Section \ref{nogo}, Bohmian theory has to be explicitly nonlocal (making it unattractive for synthesising with a causal theory of gravity). 

\item Objective representations of the measurement procedure have been described using stochastic extensions of the Schr\"{o}dinger equation \cite{Pearle:1976, Ghirardi, Percival:1995}. However, in this framework, conventional stochastic dynamics - based on Euclidean norms - would also destroy the precise `gappy' properties of the invariant set which, as discussed in Section \ref{nogo}, allow it to evade being constrained by Bell inequalities but without nonlocality.

\item As discussed, the process of measurement corresponds to the divergence and clustering of fine-scale trajectories. On the coarse-scale it appears as if the trajectory branches into two with some probabilistic weighting of alternatives (Fig \ref{decohere}b). The latter is the way Everett theory is qualitatively interpreted. However, such branching is illusory - it arises merely from taking too coarse-grain a picture of $I_U$. In addition, as discussed, the neighbouring trajectories on $I_U$ do not represent `many worlds', but are rather representations of a mono-universe at different cosmological epochs. Above all, in invariant set theory, the Schr\"{o}dinger equation is seen as a (singular) approximation to reality. As such the proposal here is very far from Everettian theory, which views quantum theory as complete. 

\item Here it is speculated that state-space clustering is a manifestation of the phenomenon of gravity. This would suggest, consistent with earlier speculations by Penrose \cite{Penrose:1989} and Di\'{o}si \cite{Diosi:1989}), that the cluster centroids are, in some sense, gravitationally distinct. Following Penrose \cite{Penrose:1989}, let $E_G$ denote the gravitational interaction energy associated with space-times $\mathcal M_1$ and $\mathcal M_2$ associated with two trajectories on $I_U$ - loosely speaking, the energy needed to deform the mass distributions $\rho_1$ in $\mathcal M_1$ to the mass distributions $\rho_2$ in $\mathcal M_2$ and definable in Newtonian theory as
\be
E_G= \frac{G}{2} \int  \frac{(\rho_1(x)-\rho_2(x))(\rho_1(y)-\rho_2(y))}{|x-y|} d^3x d^3y
\ee
or in general relativity in the weak-field limit. We will say that these space-times can be considered gravitationally indistinct over the timescale $\tau$ if
\be
\label{EG}
\int_{\tau} E_G(\mathcal M_1, \mathcal M_2) dt < \hbar
\ee
On this basis, elements of the helix of trajectory segments are gravitationally indistinct, but the cluster centroids are not. It is plausible that clustering can start to occur before the condition of gravitational distinctness has been reached, according with experimental experience that a measurement outcome may be determined before the condition for gravitational distinctness has been reached. 

\item If clustering is a manifestation of gravity, then the $1/N$ must somehow be related to the relative weakness of the gravitational force. This suggests a value for $1/N$ on the order of gravitational coupling constant (the squared mass of the electron in Planck units) $O(10^{-43})$. 

\end{itemize}

\subsection {Relation to Classical Theory}

It is of interest to ask in what sense classical theory is a limiting aspect of invariant set theory. If $N$ is kept fixed at its finite value, then as $\hbar \rightarrow 0$, the inertial properties of particles associated with a single trajectory segment in state space become decoupled from the geometry of the neighbouring trajectories. In this limit, the dynamical equations of motion can be written as $\dot X = F[X]$ (rather than, say, $\dot X=F[X, \hbar \nabla_\perp X]$ where $\nabla_\perp$ describes a type of $p$-adic derivative), i.e. the time rate of change of $X$ on a trajectory segment depends only on values of $X$ on that specific trajectory segment. This limit can be considered `quasi-classical' because, with finite $N$, the metric is still non-Euclidean, consistent with the non-classical notion of the primacy of the invariant set geometry. By putting $N=p=\infty$, then finally $g_p$ becomes Euclidean everywhere in state space, and the corresponding theory is classical. This final limit is singular: for any finite $p$, no matter how large, $g_p$ is not Euclidean on state space. Classical theory then arises as the double limit $\hbar \rightarrow 0$, $N=\infty$ of invariant set theory. As discussed above, keeping $\hbar$ fixed but putting $N=\infty$ leads to the Schr\"{o}dinger equation of quantum theory. 

\section{Negating Quantum No-Go Theorems}
\label{nogo}

Below it is concluded, surprisingly perhaps, that no experiment has demonstrated the violation of the Bell Inequality, even approximately. In invariant set theory the Bell inequality is neither satisfied nor violated: it is undefined. 

\subsection{The Bell Inequalities}

Recent experiments (e.g. \cite{Shalm}) have seemingly put beyond doubt the conclusion that the CHSH version 
\begin{equation}
\label{CHSH}
|\text{Corr}(0,0)+\text{Corr}(1,0)+\text{Corr}(0,1)-\text{Corr}(1,1)| \le 2
\end{equation}
of the Bell Inequality is violated robustly for a range of experimental protocols and measurement settings. It is very widely believed that these results show that physical theory cannot be based on Einsteinian notions of realism and local causality. Here, $\text{Corr}(X,Y)$ denotes the correlation between spin measurements (with outcomes `+' and '-'), performed by Alice and Bob on entangled particle pairs produced in the singlet quantum state
\be
|\psi\rangle=\frac{1}{\sqrt 2} (|\hat{\mathbf{X}},+\rangle|\hat{\mathbf{Y}},-\rangle-|\hat{\mathbf{X}},-\rangle|\hat{\mathbf{Y}},+\rangle
\ee
where $\hat{\mathbf{X}}$, $\hat{\mathbf{Y}}$ denote arbitrary measurement orientations and $X=0, 1$ and $Y=0,1$ are pairs of freely-chosen points on Alice and Bob's celestial spheres, respectively, thereby defining four corresponding directions.  Let $\theta_{XY}$ denote the relative orientation between an $X$ point and a $Y$ point. As discussed in Section \ref{tensorreal}, complex Hilbert tensor product states can represent the multi-variate probabilistic elements of trajectory segments on $I_U$ providing squared amplitudes are rational, and phase angles are rational multiples of $2\pi$. Hence, in invariant set theory as in quantum theory, 
\be
\text{Corr}(X,Y)=\langle \psi| (\boldsymbol{\sigma}. \hat{\mathbf{X}})(\boldsymbol{\sigma}. \hat{\mathbf{Y}})| \psi\rangle= -\cos \theta_{XY}
\ee
requiring $\cos \theta_{XY} \in \mathbb Q$. 

Of course, in the \emph{precise} form as written, (\ref{CHSH}) has not been shown to have been violated experimentally. In practice, the four correlations on the left-hand side of (\ref{CHSH}) are each estimated from a separate sub-ensemble of particles with measurements performed at different times and/or spatial locations. Hence, for example, the measurement orientation corresponding to $Y=0$ for the first sub-ensemble cannot correspond to \emph{precisely} the same measurement orientation $Y=0$ for the second sub-ensemble. As before, although Bob can try to minimise the effects of internal or external noise, as a matter of principle he cannot shield his apparatus from the effects of gravitational waves associated for example with distant astrophysical events. Hence, what is actually violated experimentally is not (\ref{CHSH}) but 
\be
\label{CHSHmod}
|\text{Corr}(0,0)+\text{Corr}(1,0')+\text{Corr}(0',1)-\text{Corr}(1',1')| \le 2
\ee
where to experimental accuracy (i.e. with respect to the Euclidean metric of space-time)$0\approx 0'$ and $1 \approx 1' $ for $X$ and $Y$. As before, could the difference between  $0\approx 0' $, $1 \approx 1' $ on the one hand, and $0 = 0'$, $1 = 1'$ on the other, actually matter? Is (\ref{CHSH}) the singular rather than the smooth limit of (\ref{CHSHmod}) as $0' \rightarrow 0$, $1' \rightarrow 1$?

Let us suppose that Alice chooses $X=0$ and Bob chooses $Y=0$ when measuring a particular entangled particle pair. Then, for this measurement, it must be the case that $\cos \theta_{00} \in \mathbb Q$.  Consider the counterfactual question: Given that Alice and Bob in fact chose $X=0$ and $Y=0$ respectively, could Alice and Bob have chosen $X=1$ and $Y=0$ respectively? The three points $X=0$, $Y=0$ and $Y=1$ are shown in Fig \ref{F:CHSH}a. We can assume that the three points are approximately coplanar, i.e. that, to experimental accuracy, the internal angle $\gamma$ in Fig \ref{F:CHSH} is approximately equal to $180^\circ$. To answer this question in the affirmative, it must be that the world $U'$ in which this counterfactual experiment takes place also lies on $I_U$. Hence it must be that $\cos \theta_{10}$ is rational. However, from the cosine rule for spherical triangles, we have 
\be 
\label{cosinerule}
\cos \theta_{10}=\cos \theta_{00} \cos \alpha_X + \sin \theta_{00} \sin \alpha_X \cos \gamma
\ee
where $\alpha_X$ is the angular distance between $X=0$ and $X=1$. Now it is always possible for Alice to send the particle which she has just measured relative to direction $X=0$, back into the measuring apparatus to be again measured in the $X=1$ direction. This second measurement corresponds to a simple (single-qubit) measurement where the state has been prepared in the $X=0$ direction and measured in the  $X=1$. Hence $\cos \alpha_X$ must be rational. Now, we also require the angle $\gamma$ to be a rational multiples of $2 \pi$ (see Section \ref{SG}). However, because of ubiquitous unshieldable gravitational waves, $\gamma \ne 180^\circ$ \emph{precisely}. Hence, as before, $\cos \theta_{01}$ is the sum of two terms, the first a rational and the second the product of three independent terms, the last being irrational. Being independent, these three terms cannot conspire to make their product rational. Hence $\cos \theta_{01}$ is the sum of a rational and an irrational and must therefore be irrational. Hence $U' \notin I_U$ and has no ontic status. Because the counterfactual question cannot be answered in the affirmative, $\text{Corr}(1,0)$ is undefined. This argument is readily generalised to show that it is never the case that all four correlations in (\ref{CHSH}) are definable - at most two of them are definable. 

\begin{figure}
\centering
\includegraphics[scale=0.3]{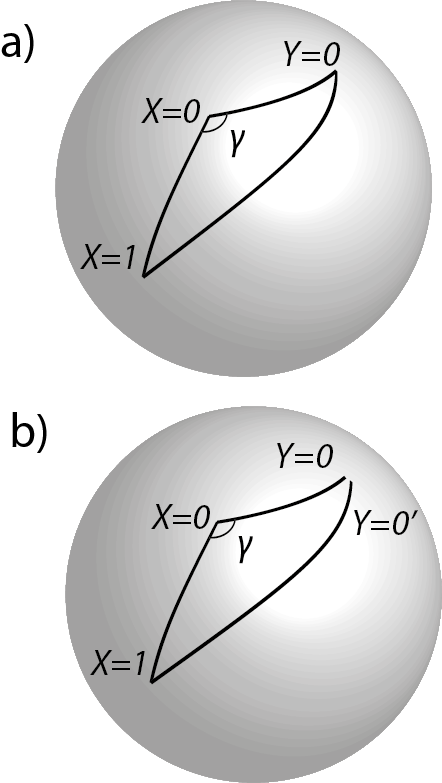}
\caption{\emph{a) In general it is impossible for all the cosines of the angular lengths of all three sides of the spherical triangle to be rational, and the internal angles rational multiples of $2\pi$. This is central to a locally realistic interpretation of the Bell inequality (\ref{CHSH}). b) What actually occurs when (\ref{CHSHmod}) is tested experimentally. It is argued that in a precise sense b) is not close to a).}}
\label{F:CHSH}
\end{figure}

In practice the first two terms on the left-hand side of (\ref{CHSHmod}) are estimated from the quasi-triangle in Fig \ref{F:CHSH}b. Here the angular lengths of all three sides have rational cosines. In a precise sense, the configuration in Fig \ref{F:CHSH}b is not $g_p$-close to Fig \ref{F:CHSH}a. (In this sense, the situation is essentially identical to the resolution of Penrose's impossible triangle.)  In conclusion, invariant set theory predicts that no physical experiment can or will be able to test (\ref{CHSH}). Physical experiments violate (\ref{CHSHmod}) but the orientations are $g_p$ far from the orientations associated with (\ref{CHSH}).  Hence in the context of physical theories with non-Euclidean metrics, experiments have not shown that the Bell inequality is violated. Hence experimental results do not conflict with invariant set theory being both realistic and locally causal. In invariant set theory the left-hand side of (\ref{CHSH}) is undefined, even though in (\ref{CHSHmod} it is well defined. Bell himself provided an argument why the violation of his inequalities was robust to noise \cite{Bell:1964}. However, this noise was defined with respect to a Euclidean state-space metric, and therefore not consistent with the primacy of $I_U$.

The same arguments apply even more straightforwardly to the original Bell inequality
\be
\label{bellorig}
|\text{Corr}(A,B)-\text{Corr}(B,C)|-\text{Corr}(A,C) \le 1
\ee
As before, application of the cosine rule for spherical triangles implies that if any two of the correlations are defined, the third is not. Hence, as before, no experiment can, as a matter of principle, test this inequality. The forms of the inequality that could, in principle, be tested experimentally are not $g_p$-approximately of the form of (\ref{bellorig}). 

In can be noted that the Kochen-Specker theorem can be negated by taking into account the inevitability that all experimentation has finite precision \cite{Meyer:1999}. Through the use of the non-classical state-space metric $g_p$ we have also shown that the Bell Theorem can similarly be negated. However, it is important to note that the violation of (\ref{CHSHmod}) does rule out conventional classical hidden-variable theories. In such theories the yardstick of state space is the Euclidean metric and with respect to this metric (\ref{CHSHmod}) can be made arbitrarily close to (\ref{CHSH}). Hence invariant set theory is not a classical theory, even though it is deterministic and locally causal. Local realism should not be confused with classicality. 

\subsection{GHZ}
\label{GHZ}

The arguments developed above can also be used to interpret the GHZ state \cite{GHZ}
\be
|\psi_{\text{GHZ}}\rangle=\frac{1}{\sqrt 2}(|v_A\rangle|v_B\rangle |v_C\rangle +|h_A\rangle|h_B\rangle |h_C\rangle)
\ee
realistically. Here a polarisation-entangled state comprising three photons is considered, where $v$ and $h$ denote vertical and horizontal polarisation. As before, given this state it is possible to make linear polarisation measurements on any of the three photons at an angle $\phi$ to the $v/h$ axis, providing $\cos^2 \phi$ and hence $\cos 2 \phi$ is rational. The corresponding unitary transformation is
\be
\label{GHZtrans}
\begin{pmatrix}
v' \\
h'
\end{pmatrix}
=
\begin{pmatrix}
\cos \phi & - \sin \phi \\
\sin \phi & \cos \phi
\end{pmatrix}
\begin{pmatrix}
v \\
h
\end{pmatrix}
\ee
It is also possible to make a circular polarisation measurement on the photons, and the corresponding unitary transformation is  
\be
\label{GHZtrans2}
\begin{pmatrix}
L\\
R
\end{pmatrix}
=
\begin{pmatrix}
1 & -i \\
1 & i
\end{pmatrix}
\begin{pmatrix}
v \\
h
\end{pmatrix}\ \ \ \ \ \ \ \ \ \ 
\ee
By considering the case $\phi \approx 45^\circ$ and both linear and circular polarisation possibilities for the photons, it is well known that it is impossible to explain measurement correlations on the GHZ state in a conventional classical local hidden-variable theory. However, it is possible to explain these correlations realistically in invariant set theory. To see this, consider, say, the second photon. Suppose in reality Bob measures this photon relative to the $v'/h'$ basis. Let us ask the counterfactual question: What would Bob have measured had he measured this photon relative to the $L/R$ basis? To answer this question, note from (\ref{GHZtrans}) and (\ref{GHZtrans2})
\begin{align}
\begin{pmatrix}
L\\
R
\end{pmatrix}
=&
\begin{pmatrix}
1 & -i \\
1 & i
\end{pmatrix}
\begin{pmatrix}
\cos \phi & \sin \phi \\
-\sin \phi & \cos \phi
\end{pmatrix}
\begin{pmatrix}
v' \\
h'
\end{pmatrix} \nonumber \\
=&
\begin{pmatrix}
e^{i \phi} & e^{i(\phi-\pi/2)} \\
e^{-i\phi} & e^{-i(\phi-\pi/2)}
\end{pmatrix}
\begin{pmatrix}
v' \\
h'
\end{pmatrix}
\end{align}
We now come to the same crucial point we arrived at in Section \ref{comp}: although $\phi$ may equal $45^{\circ}$ to any nominal accuracy, it cannot equal $45^{\circ}$ precisely because of the uncontrollable and unshieldable nature of ripples in space-time. Hence if $U$ denotes a universe where Bob chose to measure linear polarisation - implying that $\cos  2\phi$ must be rational - then the universe $U'$, identical in every way to $U$ but where Bob measures circular polarisation, cannot lie on $I_U$, because if $\cos  2\phi$ is rational, then $\phi$ cannot be a rational multiple of $2 \pi$. By construction $U'$ is $g_p$ far from $U$. Conversely, if Bob chose to measure circular polarisation, then he could not have chosen linear polarisation. That is to say, the type of counterfactual definiteness assumed in a classical hidden-variable theory is not allowed here. All the arguments about robustness to $g_p$-small noise etc follow \emph{mutatis mutandis}. 

\subsection{PBR}
\label{PBR}

The recent PBR theorem \cite{PBR} is a no-go theorem casting doubt on $\psi$-epistemic theories such as invariant set theory. Here Alice and Bob, by each choosing $0$ or $1$, prepare a quantum system in one of four input states $|\psi_0\rangle |\psi_0\rangle$, $|\psi_0\rangle |\psi_1\rangle$, $|\psi_1\rangle |\psi_0\rangle$ or $|\psi_1\rangle |\psi_1\rangle$ to some quantum circuit, where
\begin{align}
|\psi_0\rangle&=\cos \frac{\theta}{2} |0\rangle + \sin \frac{\theta}{2} |1\rangle \nonumber \\
|\psi_1\rangle&=\cos \frac{\theta}{2} |0\rangle - \sin \frac{\theta}{2} |1\rangle \nonumber
\end{align}
In addition to the parameter $\theta$, the circuit contains two further angles $\alpha$ and $\beta$. The output states of the circuit are characterised as `$\mathrm{Not}\; 00$', `$\mathrm{Not}\; 01$', `$\mathrm{Not}\; 10$' and `$\mathrm{Not}\; 00$', and $\alpha$ and $\beta$ are chosen to ensure that if Alice and Bob's input choices are $\{IJ\}$ where $I,J\in\{0,1\}$, then the probability of `$\mathrm{Not}\; IJ$' is equal to zero. The PBR theorem states that if physics is governed by some underpinning $\psi$-epistemic theory, then at least occasionally the measuring device will be uncertain as to whether, for example, the input state was prepared as $00$ and $01$. On these occasions it is possible that an outcome `$ \mathrm{Not} \; 01$' is observed when the state was prepared as $01$, contrary to quantum theory (and experiment). In invariant set theory, in an experiment on $I_U$ where the input state was $00$ and the outcome $ \mathrm{Not} \; 01$ is observed, then an experiment where the input state was $01$ and the outcome $ \mathrm{Not} \; 01$ is observed, does not lie on $I_U$. A detailed proof of this will be provided elsewhere - however, both technically and conceptually the issue is no different to that already discussed. This is consistent with the fact that the PBR theorem can be negated if a `preparation' independence assumption is violated.

\section{Invariant Set Theory and Gravity}
\label{gravity}

\subsection{Extending General Relativity onto $I_U$}

In all the iconic quantum examples above, the presence of in-principle ubiquitous unshieldable gravitational waves during non-decoherent unitary evolution has been invoked to arrive at a new locally realistic interpretation of quantum physics. This is at least consistent with the notion that invariant set theory may form the basis for a `gravitational theory of the quantum'. However, it is possible to go further. In definition \ref{EP}, the de Broglie relations were interpreted as providing a link between energy-momentum in space time, and the helical geometry of $I_U$. Since energy-momentum is the source for space-time curvature, the definition \ref{EP} effectively describes a link between the fractal geometry of state-space, and the (pseudo-) Riemannian geometry of space time. A key question arises: How should the field equations
\be
\label{gr}
G_{\mu \nu} = \frac{8 \pi G}{c^4} T_{\mu \nu} 
\ee
of general relativity generalise for finite $N$? Based on the discussion above, one can imagine generalising (\ref{gr}) so that the Einstein tensor couples not only to energy-momentum in our space-time, but also to energy-momentum on neighbouring space-times on $I_U$. Consider an isolated self-gravitating system comprising some finite number $M$ of individual particles, then it seems plausible to generalise (\ref{gr}) as 
\be
\label{gr2}
G_{\mu \nu}(\mathcal M_{U_o})= \frac{8 \pi G}{c^4} \int_{U \in I_U} T_{\mu \nu}(\mathcal M_{U}) \Delta(\mathcal M_{U_0}, \mathcal M_{U}) \ d\mu
\ee
where $U_o$ denotes the observed universe in which we live, $\mu$ is the natural Haar measure on $I_U$, and $\Delta$ is a non-singular kernel over the associated $3M$ complex dimensional state space (that is to say $3M$ dimensions of helices), with the property 
\be
\label{normalisation}
Q_M \equiv \int_{U \in I_U} \Delta(\mathcal M_{U_0}, \mathcal M_{U}) d \mu >1
\ee 
such that
\be
\lim_{N \rightarrow \infty}\Delta(\mathcal M_{U_0}, \mathcal M_{U})=\delta(\mathcal M_{U_0}, \mathcal M_{U}) \nonumber
\ee
a Dirac delta function, implying $Q_M \rightarrow 1$ as $N \rightarrow \infty$. It can be noted that (\ref{gr}) is the smooth and not the singular limit of (\ref{gr2}) as $N \rightarrow \infty$, lending further support to the notion of invariant set theory as potentially closer to gravitation theory than to quantum theory. Because the generalisation from (\ref{gr}) to (\ref{gr2}) only directly affects the source term of the Einstein tensor in $\mathcal M_{U_o}$ and not the Einstein tensor itself, this generalisation does not disturb the causal structure in  $\mathcal M_{U_o}$, consistent with the claim that invariant set theory is locally causal. The precise form of $\Delta$ is not defined here. Nevertheless, it is possible to speculate on some of its implications: 
\begin{itemize}

\item If $Q_M$ increases monotonically with $M$, then for large enough $M$ (corresponding, perhaps, to a galaxy), the Einstein tensor in $\mathcal M_{U_o}$ could be substantially influenced by the distribution of energy-momentum in neighbouring space times on $I_U$. If we attempt to rewrite (\ref{gr2}) as
\be
\label{gr3}
G_{\mu \nu} = \frac{8 \pi G}{c^4} (T_{\mu \nu} + \Delta T^{\text{eff}}_{\mu\nu})
\ee
then the additional source term $\Delta T^{\text{eff}}_{\mu\nu}$ for the Einstein tensor in $\mathcal M_{U_o}$ would appear to describe some distribution of `cold dark matter' in $\mathcal M_{U_0}$ itself. 

\item From (\ref{gr2}) one might also expect $G_{\mu\nu}(\mathcal M_{U_o})$ to feel the effects of the ubiquitous divergence of trajectories, c.f. Fig \ref{decohere}, on the largest scales $M \rightarrow \infty$, suggesting the form $\Delta T^{\text{eff}}_{\mu\nu}= \Lambda g_{\mu\nu}$ for (\ref{gr3}) where $\Lambda>0$. From this perspective, it can be understood why quantum vacuum fluctuations may not themselves be a significant source for `dark energy'. In invariant set theory, one can view vacuum fluctuations as representing variations across $j+1$th iterate trajectories of $I_U$, relative to some $j$th iterate trajectory. When computing the energy density associated with the vacuum, one should not integrate up to Planckian frequencies, but only to energies above which the fine-scale trajectories start to decohere and cluster - tiny by comparison. From this perspective, a typical decoherence timescale in the universe could play a vital role in determining the magnitude of dark energy. For example, in the early hot universe where typical decoherent times could be much faster than they are today, the acceleration of the universe would be predicted to be correspondingly larger. This may conceivably be relevant for understanding the origin of cosmic inflation. In the late universe, dominated by state-space velocity convergence (see below), invariant set theory predicts that $\Lambda < 0$.  

\item The finite kernel $\Delta$ may plausibly prevent the occurrence of singularities in gravitational collapse. In neighbouring space-times on $I_U$, the gravitational collapse of some star, for example, would proceed in different ways, each with its own $T_{\mu\nu}$. These differences, when actively coupled to $G_{\mu\nu} (\mathcal M_{U_o})$ through the kernel $\Delta$, may plausibly be sufficient to prevent the Einstein tensor becoming singular. The elimination of such singularities has always been seen to be a key requirement of any putative theory which synthesises quantum and gravitational physics. 

\item In conventional physical theory, time is a parameter describing a foliation of spacelike hypersurfaces in space-time (time goes `up the page'). A notoriously unsolvable problem associated with this perspective is that of defining the notion of `now'. In particular, there is nothing in conventional physical theory that can explain objectively the visceral sensation that the present is somehow different from the past or the future. This problem is potentially solvable using the finite kernel $\Delta$. In invariant set theory, time can be parametrised by the fractal iteration number $j$ (see Fig \ref{decohere}). In this perspective, the passage of time corresponds to a `zooming' into the fractal - one might recall the remarkable animations which zoom into the Mandelbrot set. By way of illustration consider the fractal structure in Fig \ref{p5fractal} and let $\Delta_5$ denote a kernel which averages over one of the five big diamonds. Then the future is illustrated by one of the smaller diamonds (and the future to the future by a yet higher-iterate diamond, within a small-scale diamond, not illustrated). That is to say, in invariant set theory, the future is associated with parts of $I_U$ that are gravitationally indistinguishable. By contrast, the past defines fractal iterates $j$ which have expanded to become much larger than the effective domain of $\Delta$, i.e. in invariant set theory these spaces have become gravitationally independent. For example in Fig \ref{p5fractal}, one can imagine that the structure illustrated is itself one of five elements of some larger diamond, not shown. Through $\Delta$, invariant set theory provides a solution to the problem of what makes `now' unique. 
\end{itemize}

\subsection{Fractal Gaps and Information Compression}

As discussed, the existence of fractal gaps in $I_U$ are crucial for this realistic interpretation of quantum physics. What the implications of such gaps are for gravitation theory? In classical dynamical systems theory, the fractal structure of strange attractors is underpinned by three basic ingredients: instability, nonlinearity and dissipation. Of these, dissipation is a phenomenological and not fundamental process. Dissipation implies that state-space volumes typically decrease through a convergence $\nabla . \bf v <0$ of the state-space velocity field $\bf v$ under dynamical evolution. The question therefore arises - is there a fundamental process in gravitation theory that would imply the existence of regions in state space where $\nabla . \bf v <0$? As discussed by Penrose (see Fig 3.14 of \cite{Penrose:2010} in particular), the black-hole no-hair theorem \cite{Israel} can be interpreted as implying a convergence of the cosmological state-space velocity field. It is enough for $\nabla . \bf v <0$ to only occur in relatively small neighbourhoods of state space for the geometry of $I_U$ to have fractal gaps everywhere. By way of illustration, consider a classical deterministic system whose dynamics are Hamiltonian ($\nabla . \bf v =0$ in state space), apart from in a small state-space volume $V_{\text{conv}}$ where $\nabla . \bf v <0$. Providing trajectories pass through $V_{\text{conv}}$ from time to time, then the system's invariant set can exhibit fractal (i.e. `gappy') structure everywhere, even though the dynamical equations are Hamiltonian in most regions of state space. That is to say, global properties of the invariant set can be inherited from properties localised to sub-regions of state space. This argument highlights the notion that an invariant set is a profoundly holistic and integrative construct in state space and why invariant set theory differs from more reductionist approaches to fundamental physics.  It can be noted that whilst Penrose \cite{Penrose:2010} equates $\nabla . \bf v <0$ with information loss, since the volume of $I_U$ is zero, volumes cannot shrink on $I_U$ and hence information is instead becoming more and more compressed in regions of state space where trajectories are converging. That is to say, in invariant set theory `information' is never actually lost, even though it may seem to be. 

\subsection{Observational Consequences}

It is possible to speculate about some observational consequences of invariant set theory:

\begin{itemize}

\item Invariant set theory predicts that though the state-space clustering process, gravity is an inherently decoherent process. For example, the gravitational field will not itself be capable of being an `entanglement witness'. Experiments are now being designed to study this \cite{Marletto}.   

\item Invariant set theory predicts that since gravity is not a quantum field in space time, there is no such particle as a graviton, consistent with the conclusions of Dyson \cite{Dyson:2004}

\item If the dark universe is associated with a generalisation of general relativity as discussed above, it will not be associated with excitations of quantum fields in space time and therefore won't be associated with supersymmetric particles in particular. As such, it is predicted that no elementary particle exists with spin greater than one.   

\end{itemize}

\section{Conclusions and Further Developments}
\label{conclusions}

In 1976, Roger Penrose wrote \cite{Penrose:1976}:

\begin{quote}
`Despite impressive progress \ldots towards the intended goal of a satisfactory quantum theory of gravity, there remain fundamental problems whose solutions do not appear to be yet in sight. \ldots [I]t has been argued that EinsteinÕs equations should perhaps be replaced by something more compatible with conventional quantum theory. There is also the alternative possibility, which has occasionally been aired, that some of the basic principles of quantum mechanics may need to be called into question.' 
\end{quote}
This quote is as relevant today as it was over 40 years ago. Invariant set theory is a proposal for what this `alternative possibility' might be. Perhaps the most important of quantum theory's basic principles that has been called into question is that the state space of quantum physics is an algebraically closed space (Hilbert Space) based on the continuum field $\mathbb C$. Since $\mathbb C$ is itself the algebraically closed extension of $\mathbb R$, and $\mathbb R$ is the completion of $\mathbb Q$ with respect to the familiar Euclidean metric, then this in turn calls into question whether the Euclidean metric is necessarily the correct physical yardstick in state space. From number theory it is known that there exists a uniquely non-Euclidean metric, the $p$-adic metric, which is to fractal geometry as the Euclidean metric is to Riemannian geometry. Invariant set theory, then, is a theory based on the assumption that the universe is a deterministic locally causal system evolving on a fractal-like geometry $I_U$ in state space, and with an assumed metric $g_p$ (related to the $p$-adic metric) which respects the primacy of $I_U$. Here $p$, a Pythagorean prime whose inverse magnitude reflects the relative weakness of gravity, is a finite parameter describing a helical fractal structure for trajectory segments on $I_U$. Invariant set theory is essentially a finite theory: when descriptors of putative (typically counterfactual) Hilbert states need irrationals for their description, then such states cannot represent elements of physical reality, i.e. cannot lie on $I_U$, and must rather lie in one of the fractal gaps associated with $I_U$. This property is central to a realistic account of quantum complementarity and the non-commutativity of spin operators, and makes invariant set theory an attractive feature compared with conventional theories based on $\mathbb R$ and $\mathbb Q$ (where it is impossible for states with rational and irrational descriptors to be physically distinguishable). Here, as discussed, the transcendental nature of the cosine function plays a critical role. 

As has been shown, quantum theory arises as a singular \cite{Berry} limit of invariant set theory in the non-gravitational limit when $p=N+1$ is set equal to $\infty$. By contrast, Einstein's theory of general relativity can arises smoothly as $p \rightarrow \infty$. This singular limiting behaviour is central to showing that no experiment has demonstrated, or will ever demonstrate, that the Bell Inequalities are violated, even $g_p$-approximately. Hence there is no basis in experiment to reject Einstein's belief that physical theory should be both deterministic and locally causal. Although locally causal and deterministic, invariant set theory is not a classical theory - the fact that $g_p$ is not Euclidean necessarily makes it so. However, at its deepest, invariant set theory departs from conventional theory (both classical and quantum) in being holistic rather than reductionist. In conventional theory, large-scale cosmology is presumed to arise - bottom up -  from the basic building blocks of a putative quantum theory of gravity (electrons, photons, gravitons and so on). Here, by contrast, a description of the smallest quantum building blocks of nature have been sought from the state-space geometry of the largest conceivable gravitationally-bound structure, the state-space geometry $I_U$ of the universe as a whole. Hence it is concluded that the elements of invariant set theory presented here - as much top down as bottom up \cite{Ellis} - may form the basis for a new approach to synthesising quantum and gravitational physics: a gravitational theory of the quantum rather than a quantum theory of gravity. In such a synthesis there will be no such particle as a graviton. 

Clearly further developments are needed to turn invariant set theory into a complete and rigorous theory of quantum and gravitational physics. Firstly, it is necessary to recast quantum field theory in terms of fractal state-space geometry. This raises the question of how the gauge fields of the Standard Model are representable in terms of the helical geometry of trajectories on $I_U$, generalised through the quaternions (\ref{quaternions}) which are known to underpin a description of the gauge field theory. A possible route forward could be through the theory of $p$-adic manifolds as Lie Groups \cite{Schneider}. The second development will be a demonstration of the emergence of the Lorentzian pseudo-Riemannian geometry of space-time from the quaternionic structure of $I_U$. Finally, the global geometry of $I_U$ needs to be defined and analysed. The author hopes to present such developments in due course. 

\bigskip
\bibliography{mybibliography}
\end{document}